\newif\ifjournal
\journalfalse
\documentclass{article}
\usepackage{graphicx}
\usepackage{graphics}
\usepackage{moreverb}
\usepackage{subfigure}
\usepackage{lp}
\usepackage{latexsym}
\usepackage{amsfonts}
\usepackage{url}
\usepackage{amssymb}
\usepackage{amscd}
\usepackage{amsbsy}
\usepackage{amssymb}
\usepackage{amsmath}
\usepackage{amsthm}
\usepackage{hyperref}
\usepackage{fullpage}

\newtheorem{theorem}{Theorem}[section]
\newtheorem{lemma}[theorem]{Lemma}

\newtheorem{corollary}[theorem]{Corollary}

\newenvironment{proofof}[1]{\smallskip\par\noindent{\em Proof of #1.}}{\qed}
\newcommand{\set}[1]{\{ #1 \}}
\newcommand{\algorithm}[2]
{\small
{\fbox{\fbox{\begin{minipage}{4in}
{\bf #1}\vspace*{.1cm}\hrule
\begin{tabbing}
\qquad\=\qquad\=\qquad\=\qquad\=\qquad\=\qquad\=\kill
#2
\end{tabbing}
\end{minipage}}}}}

\begin{document}

\title{Online Constrained Forest and Prize-Collecting Network Design}

\author{Jiawei~Qian%
\thanks{Ping An Securities, Hong Kong SAR. 
Email:
\href{mailto:jq35@cornell.edu}{jq35@cornell.edu}.
Research performed while at Cornell University and supported in part by NSF grant CCF-1115256.}
\and
Seeun~William~Umboh\thanks{Department of Mathematics and Computer Science, Eindhoven University of Technology, 5600 MB Eindhoven, The Netherlands.
Email:
\href{mailto:seeun.umboh@gmail.com}{seeun.umboh@gmail.com}.
Part of this work was done while visiting the Simons Institute for the Theory of Computing and supported by NWO Vidi grant 639.022.211.}
\and
David~P.~Williamson\thanks{School of Operations Research and Information Engineering, Cornell University, Ithaca, NY, 14853, USA.
Email:
\href{mailto:dpw@cs.cornell.edu}{dpw@cs.cornell.edu}.
Supported in part by NSF grant CCF-1115256.}
}

\date{\today}

\maketitle

\begin{abstract}
In this paper, we study a very general type of online network design problem, and generalize two different previous algorithms, one for an online network design problem due to  Berman and Coulston \cite{BermanC97} and one for (offline) general network design problems due to Goemans and Williamson \cite{GoemansW95}; we give an $O(\log k)$-competitive algorithm, where $k$ is the number of nodes that must be connected.  We also consider a further generalization of the problem that allows us to pay penalties in exchange for violating connectivity constraints; we give an online $O(\log k)$-competitive algorithm for this case as well.  
\end{abstract}

\ifjournal
\fi

\section{Introduction}

Network design has been a fundamental application of techniques in combinatorial optimization for some time; see the volume of Ball et al.\ \cite{BallMMN95} for an overview.  Most models assume that all the connectivity requirements are given in advance.  However, it is sometimes the case that decisions in constructing the network must be made as customers arrive over time; decisions to build network infrastructure must be made at the time the customer arrives, and cannot be undone in later time steps.  Such problems have been studied under a model known as {\em online} decision making; algorithms in this model are measured in terms of their {\em competitive ratio}, which gives a bound on how far away the algorithm's solution can be away from an optimal solution found when given all the connectivity information in advance.  Problems in which all the input (including connectivity information) is known in advance are then called {\em offline} problems.

As a running example, we define here the {\em generalized Steiner tree problem}, also known as the {\em Steiner forest problem}.  In the offline version of this problem, we are given an undirected graph $G=(V,E)$, edge costs $c_e \geq 0$ for all $e \in E$, and a set of $k$ source-sink pairs $s_i$-$t_i$ as input. The goal of the problem is to find a minimum-cost set of edges $F \subseteq E$ such that for each $i$, $s_i$ and $t_i$ are connected in $(V,F)$.  This problem is (as its name implies) a generalization of the {\em Steiner tree problem}: in the offline version of the Steiner tree problem, we are given an undirected graph with edge costs as above, and also a set $R \subseteq V$ of {\em terminals}.  The goal of the Steiner tree problem is to find a minimum-cost tree $T$ that spans all the terminals in $R$.
The Steiner tree problem is one of Karp's original NP-hard problems \cite{Karp72}.
If we choose one of the terminals $r \in R$ arbitrarily, set $s_i = r$ for all $i$, and let the sink vertices $t_i$ be the remaining vertices in $R$, then clearly a Steiner tree instance can be expressed as a generalized Steiner tree problem instance.

In the online version of the generalized Steiner tree problem, we do not know the source-sink pairs in advance.  The online problem proceeds in a sequence of discrete time steps; in each time step $i$, a source-sink pair $s_i$-$t_i$ arrives, and we must find a set of edges $F$ such that each $s_j$-$t_j$ pair that has arrived thus far is connected in $(V,F)$.  Furthermore, once we have decided to include an edge in $F$, we may not remove it at later time steps; once we have constructed an edge in our network, the cost is sunk and we may not recover it at future points in time.  

The following simple example shows that in an online setting, we cannot in general find an optimal offline solution, even given unlimited computational power. Consider the 4-cycle with vertices $v_1,v_2,v_3,v_4$ in Figure \ref{online_ex}. All edges have cost 1.  
Suppose $(v_1, v_3)$ is the first pair to arrive, in time step 1. We can choose either path $(v_1,v_2,v_3)$ or path $(v_1,v_4,v_3)$ to connect it. Without loss of generality, we will choose path $(v_1,v_2,v_3)$. Then, if $(v_1,v_4)$ arrives in the second time step, we could have saved a cost of one if we had chosen the other path in the first time step. However, even if we did that, $(v_1,v_2)$ could be the pair arriving at time step 2 and we would face the same problem.

\begin{figure}[t]
  \centering
     \includegraphics[width=0.3\textwidth]{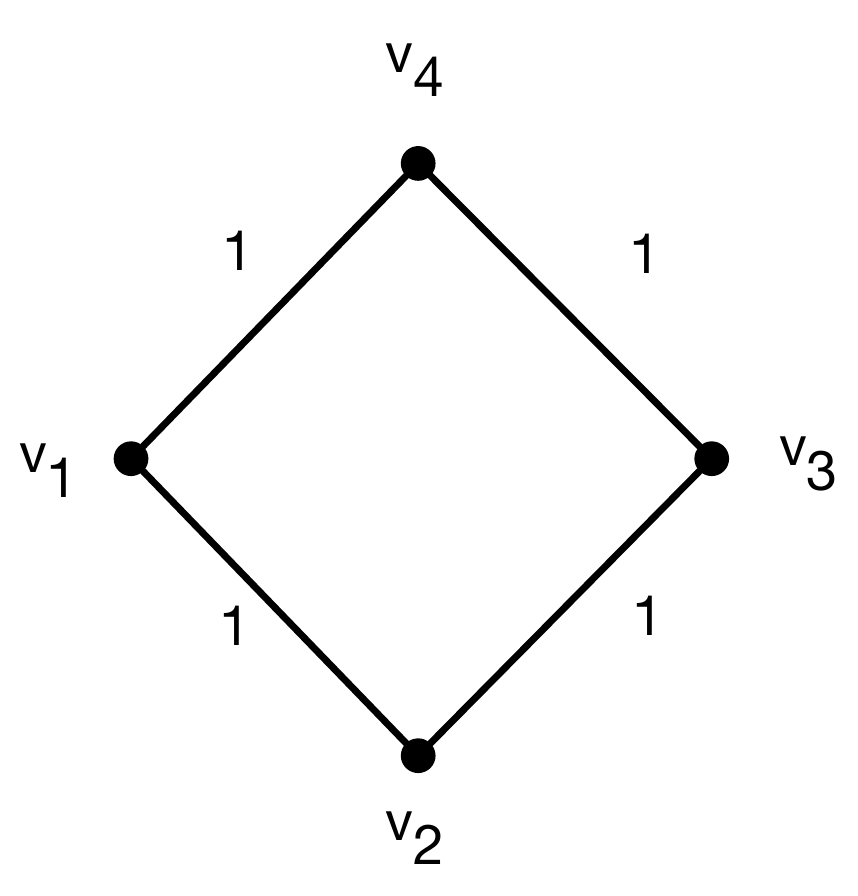}\\
  \caption{Example of the online generalized Steiner tree problem}\label{online_ex}
\end{figure}

As mentioned above, the quality of an online algorithm is often measured in terms of its {\em competitive ratio}: an {\em $\alpha$-competitive algorithm} is one such that at any time step, the value of current solution is within a factor of $\alpha$ of the value of an optimal offline solution.  For the online generalized Steiner tree problem, an $\alpha$-competitive algorithm constructs a set of edges that at the current time step has cost at most $\alpha$ times the cost of the optimal solution for the set of source-sink pairs that have arrived thus far.  This notion should be compared to that of an approximation algorithm.  Approximation algorithms are given for offline problems; an $\alpha$-approximation algorithm is guaranteed to run in polynomial time and produce a solution with cost at most $\alpha$ times the value of an optimal solution. Agrawal, Klein, and Ravi \cite{AgrawalKR95} give a 2-approximation algorithm for the offline generalized Steiner tree problem.

Online algorithms are known for both the online Steiner tree problem and the online generalized Steiner tree problem.  In the online version of the Steiner tree problem, terminals arrive over time.  At each time step we must give a set of edges $F$ that connects all of the terminals that have arrived thus far; we are not allowed to remove any edges from $F$ in future iterations. As stated above, in the online generalized Steiner tree problem, source-sink pairs arrive in each time step, and we must find a set of edges $F$ such that each $s_i$-$t_i$ pair that has arrived thus far is connected in $(V,F)$. Imase and Waxman \cite{ImaseW91} give a greedy $O(\log k)$-competitive algorithm for the online Steiner tree problem, where $k$ is the number of terminals; when a terminal arrives, it finds the shortest path from the terminal to the tree already constructed, and adds that set of edges to its solution.  Imase and Waxman also show that the competitive ratio of any online algorithm must be at least $\frac{1}{2} \log_2 k$; one can show this lower bound by repeatedly replacing each edge in the graph of Figure \ref{online_ex} with a copy of the graph. Awerbuch, Azar, and Bartal \cite{AwerbuchAB04} show that a similar greedy algorithm for the online generalized Steiner tree problem has a competitive ratio of $O(\log^2 k)$.   Berman and Coulston \cite{BermanC97} give a more complicated algorithm that is an $O(\log k)$-competitive algorithm for the online generalized Steiner tree problem, matching the lower bound of Imase and Waxman to within constant factors.

Part of the contribution of this paper is to extend the types of network design problems for which online algorithms are known. Goemans and Williamson \cite{GoemansW95} extended the offline algorithm of Agrawal, Klein, and Ravi \cite{AgrawalKR95} to a large class of problems they called constrained forest problems; in doing so, they cast the algorithm of Agrawal et al.\ as a primal-dual algorithm, one that simultaneously constructs a feasible primal solution to an integer programming formulation of the problem as well as a feasible solution to the dual of a linear programming relaxation.  A constrained forest problem is defined by a function $f:2^V \rightarrow \{0,1\}$; for any set $S \subseteq V$ such that $f(S) = 1$, a feasible solution must select at least one edge in $\delta(S)$, the set of edges with exactly one endpoint in $S$.  The Goemans-Williamson algorithm works when the function $f$ is {\em proper}: that is, when $f(S) = f(V-S)$ for all $S \subseteq V$, $f(\emptyset) = f(V) = 0$ and for all disjoint sets $A, B \subseteq V$, $f(A \cup B) \leq \max(f(A), f(B))$; we also assume that $f$ is polynomial-time computable. For instance, for the case of the generalized Steiner tree problem $f(S) = 1$ if and only if there exists some $i$ such that $|S \cap \{s_i,t_i\}| = 1$, and this function is proper.  Another example of a constrained forest problem given in Goemans and Williamson \cite{GoemansW95} is the nonfixed point-to-point connection problem, in which a subset $C$ of vertices are sources, a disjoint subset $D$ of vertices are destinations, and we must find a minimum-cost set of edges such that each connected component has the same number of sources and destinations; this is modelled by having $f(S) = 1$ if $|S \cap C| \neq |S \cap D|$.  Yet another example given in \cite{GoemansW95} is that of partitioning specified vertices $D$ into connected components such that the number of vertices of $D$ in each connected component $C$ is divisible by some parameter $\ell$.  This problem is given the proper function $f$ such that $f(S) = 1$ if $|S \cap D| \not\equiv 0(\mathrm{mod~}\ell)$.

In this paper, we show that by melding the ideas of Goemans and Williamson \cite{GoemansW92} with those of Berman and Coulston \cite{BermanC97}, we can obtain an $O(\log k)$-competitive algorithm for any {\em online} constrained forest problem.  In an online constrained forest problem, in each time step $i$ we are given a proper function $f_i$.  We must choose a set of edges $F$ such that for all $S \subseteq V$, if $\max_{j = 1,\ldots,i} f_j(S)=1$, then $|\delta(S) \cap F| \geq 1$ (one can verify that the function $\max_{j = 1,\ldots,i} f_j(S)$ is itself proper).  In our case, $k$ is the number of vertices $v$ for which $f_i(\{v\})=1$ for some $i$.  This yields, for example, algorithms for online variants of the nonfixed point-to-point connection problem and the partitioning problem given above.

Our techniques also extend to give an $O(\log k)$-competitive algorithm for a very general set of network design problems in which we may wish to pay a penalty instead of fulfilling a connectivity requirement. 
One such example is that of the prize-collecting Steiner tree problem. In the offline version of the prize-collecting Steiner tree problem, we are given an undirected graph $G=(V,E)$, edge costs $c_e \geq 0$ for all $e \in E$, a root vertex $r \in V$, and penalties $\pi_v \geq 0$ for all $v \in V$.  The goal is to find a tree $T$ spanning the root vertex that minimizes the cost of the edges in the tree plus the penalties of the vertices not spanned by the tree; that is, we want to minimize $\sum_{e \in T} c_e + \sum_{v \in V-V(T)} \pi_v$, where $V(T)$ is the set of vertices spanned by $T$.  In the online version of the problem, initially every vertex $v$ has penalty $\pi_v = 0$.  At each time step, the penalty $\pi_v$ for some vertex $v$ is increased from 0 to some positive value.  We then must either connect the vertex to the root by adding edges to our current solution or pay the penalty $\pi_v$.  The competitive ratio of the algorithm compares the cost of our solution in each time step with the cost of the optimal solution of the instance at the same time step.  The offline version of this problem was studied by researchers at AT\&T since the problem models that of making decisions of when to extend the current network to new clients, where each penalty represents the profits forgone by not connecting the client; see  Johnson, Minkoff, and Phillips \cite{JohnsonMP00}.  Our techniques further extend to online versions of the prize-collecting generalized Steiner tree problem introduced by Hajiaghayi and Jain \cite{HJ06}. The online prize-collecting generalized Steiner tree problem is as follows: initially we are given an undirected graph $G$, and a penalty of zero for each pair of nodes. In each time step $i$, a terminal pair $(s_i, t_i)$ arrives with a new penalty $\pi_i >0$. We have a choice to either connect $s_i$ to $t_i$ or pay the penalty $\pi_i$ for not connecting them.  Our goal is to find a set of edges $F$ that minimizes the sum of edge costs in $F$ plus the sum of penalties for terminal pairs that are not connected.  Our technique also extends to an online version of a problem of Hayrapetyan, Swamy, and Tardos \cite{HST05}, in which we must minimize the cost of a tree spanning a root vertex $r$, plus a monotone submodular penalty function $h$ on all the unspanned vertices. In the online version, in each time step $i$, a new monotone submodular function $h_i$ arrives.   See Section \ref{sec:pccf} for more details.  We obtain our results by giving an $O(\log k)$-competitive algorithm for an online version of the prize-collecting constrained forest problem introduced by Sharma, Swamy, and Williamson \cite{SSW07}, which generalizes the online prize-collecting Steiner tree problem, the online prize-collecting generalized Steiner tree problem, and the online version of the problem of Hayrapetyan et al.  We introduce this general problem in Section \ref{sec:pccf}.

We now give a sketch of the algorithmic ideas and the analysis.  The basic idea of the Berman-Coulston algorithm (BC) is that it constructs many different families of nonoverlapping balls around terminals as they arrive; in the $j$th family, balls are limited to have radius at most $2^j$.  Each family of balls is a lower bound on the cost of an optimal solution to the generalized Steiner tree problem; the balls can be seen as a feasible solution to the dual of a linear programming relaxation of the problem.  When balls from two different terminals touch (corresponding to a tight dual constraint), the algorithm buys the set of edges connecting the two terminals, and balls from one of the two terminals (in some sense the `smaller' one) can be charged for the cost of the edges, leaving the balls from the other terminal (the `larger' one) uncharged and able to pay for future connections.  Thus by induction, it can be shown that the cost of the edges constructed can be charged to the balls in all the families.  One can show that the $O(\log k)$ largest families are essentially all that are relevant for the charging scheme, so that the largest of these $O(\log k)$ families is within an $O(\log k)$ factor of the cost of the constructed solution, thereby giving the competitive ratio. Our algorithm for the online constrained forest problem extends the BC algorithm in several ways. First, our algorithm explicitly uses solutions to the dual of the linear programming relaxation of the constrained forest problem, as used by Goemans and Williamson, resulting in somewhat more complicated dual solutions than the balls used by BC. Second, to ensure that we output a feasible solution, our algorithm uses a more sophisticated dual construction in which the $j$th dual solution also takes into account edges that were added due to tight constraints of the other dual solutions.
In particular, our algorithm incorporates a ``consolidate'' step which ensures that the algorithm only raises dual variables that correspond to a union of a collection of connected components of $F$. However, we can then largely follow the outline of the BC analysis to obtain our $O(\log k)$ competitive ratio. 

The rest of this paper is structured as follows.  In Section \ref{sec:prelim}, we introduce the online constrained forest problem more precisely and define some concepts we will need for our algorithm.  In Section \ref{sec:alg}, we give the algorithm and its analysis.  In Section \ref{sec:pccf}, we extend the algorithm to handle penalties, and explain how the extension captures online versions of the  prize-collecting Steiner tree and prize-collecting generalized Steiner tree problem.
We conclude in Section \ref{sec:conc} with some open questions.

The online constrained forest problem and online prize-collecting Steiner tree were introduced in a preliminary version of this paper \cite{QianW11}.  However, the algorithm and analysis in this preliminary version were later discovered to be flawed and we give a corrected version of the algorithm and proofs in Section \ref{sec:alg}. Since the preliminary version appeared, there has been some additional work done on these problems and related ones. Umboh \cite{Umboh15} gives a new and simpler analysis of the Berman-Coulston algorithm for online generalised Steiner tree via the idea of hierarchically well-separated trees. He also gives another $O(\log k)$-competitive algorithm for the prize-collecting version that is analysed in the same way. For the more general node-weighted setting, in which costs are associated with nodes rather than edges, Hajiaghayi, Liaghat, and Panigrahi give polylogarithmic-competitive algorithms for the online constrained forest problem \cite{HajiaghayiLP13} and the online prize-collecting generalised Steiner tree problem \cite{HajiaghayiLP14}. For the edge-weighted setting, the algorithm of \cite{HajiaghayiLP13} yields a $O(\log k)$-competitive algorithm for the online constrained forest problem that is different from ours, and \cite{HajiaghayiLP14} also gives an alternate $O(\log k)$-competitive algorithm for the online prize-collecting Steiner tree problem. Because the preliminary version of this paper \cite{QianW11} was flawed, the paper of Hajiaghayi, Liaghat, and Panigrahi \cite{HajiaghayiLP13} had the first correct $O(\log k)$-competitive algorithm for the online constrained forest problem, and their paper \cite{HajiaghayiLP14} had the first correct $O(\log k)$-competitive algorithm for the prize-collecting Steiner tree problem.  To the best of our knowledge, there is no previous work that tackles the online prize-collecting constrained forest problem. 

\section{Preliminaries}
\label{sec:prelim}
Recall that a function $f:2^V \rightarrow \{0,1\}$ is {\bf proper} if $f(S)=f(V-S)$ for all $S \subseteq V$, $f(\emptyset)=f(V)=0$, and for disjoint sets $A, B \subseteq V$, $f(A \cup B) \leq \max(f(A),f(B))$.  Given an undirected graph $G=(V,E)$, edge costs $c_e\geq 0$ and a proper function $f$, the offline constrained forest problem studied in Goemans and Williamson \cite{GoemansW95} is to find a set of edges $F$ of minimum cost that satisfies a connectivity requirement function $f : 2^V \rightarrow \{0,1\}$; the function is {\bf satisfied} if for each set $S\subseteq V$ with $f(S)=1$, we have $|\delta(S) \cap F| \geq 1$, where $\delta(S)$ is the set of edges with exactly one endpoint in $S$.  In the online version of this problem, we have a sequence of connectivity functions $f_1,f_2,...,f_i$, arriving one at a time. Starting with $F=\emptyset$, for each {\bf time step} $i \geq 1$, function $f_i$ arrives and we need to add edges to $F$ to satisfy function $f_i$.  Once an edge is added to $F$, it cannot be removed in a later time step.  Let $g_i(S) = \max \{f_1(S),...,f_i(S)\}$ for all $S\subseteq V$ and $i \geq 1$.  Then our goal is to a find a minimum-cost set of edges $F$ that satisfies function $g_i$, that is, all connectivity requirements given by $f_1,...,f_i$ that have arrived thus far.  We require that each function $f_i$ be a proper function, as defined above.
It is easy to see that function $g_i$ is also proper.

Call a vertex $v$ a {\bf terminal} at time $i$ if $g_i(\{v\}) = 1$.  Let $R_i = \{s \in V \;|\; g_i(\{s\}) =1\}$ be the set of terminals defined by function $g_i$; that is, $R_i$ is the set of all terminals that have arrived by time $i$.  A special case of this problem is the online generalized Steiner tree problem, in which terminal pairs $(s_1,t_1),...,(s_i,t_i)$ arrive one at a time. In this case, $f_i(S) = 1$ iff $|S\cap\{s_i,t_i\}| = 1$ and $(s_i, t_i)$ is the pair of terminals that arrive in time step $i$; then $R_i = \{s_j, t_j: j \leq i\}$.  Berman and Coulston \cite{BermanC97} give an $O(\log |R_i|)$-competitive algorithm for the online generalized Steiner tree problem.

Let $(IP_i)$ be an integer program corresponding to the online proper constrained forest problem with set of functions $f_1,...,f_i$ that have arrived thus far and the corresponding function $g_i$. The integer programming formulation is
\lps
& &
& \mbox{Min} & \sum_{e \in E} c_{e} x_{e} \\
(IP_i) & & & & \sum_{e \in \delta(S)}x_{e} \geq g_i(S), & \forall S\subseteq V,\\
& & & & x_e \in \set{0,1}, & \forall e \in E. \elps

We let $(LP_i)$ denote the corresponding linear programming relaxation in which the constraints $x_e \in  \{0,1\}$ are replaced with  $x_e \geq 0$. The dual of this linear program, $(D_i)$, is
\lps & &
& \mbox{ Max} &  \sum_{S \subseteq V} g_i(S)y_S  \\
(D_i) & & & & \sum_{S: e \in \delta(S)} y_{S} \leq c_e, & \forall e \in E,\\
& & & & y_S \geq 0, & \forall S \subseteq V. \elps

We now define a number of terms that we will need to describe our algorithm.  We will keep an infinite number of feasible dual solutions $y^j$, $j = \ldots, -2,-1,0,1,2,\ldots$, to bound the cost of edges in our solution $F$ over all time steps; we call $y^j$ the dual solution for {\bf level } $j$. For each level $j$, we will maintain that for any terminal $s$ that has arrived thus far, $\sum_{S \subseteq V: s \in S} y_{S}^j \leq 2^j$. So we say that the {\bf limit} of the dual in level $j$ is $2^j$, and we say that a dual variable $y_S^j$ {\bf reaches its limit} if the inequality for level $j$ is tight for any terminal $s\in S$.  An edge $e\in E$ is  {\bf tight in level $j$} for dual vector $y^j$ if the corresponding constraint in dual problem $(D_i)$, $\sum_{S:e\in\delta(S)} y_S^{j}  \leq   c_e$, holds with equality.

Let $\bar{F}^{j}$ denote the set of edges that are tight in level $j$ plus the set of edges in the current solution $F$. To avoid confusion with connected components in $F$, we will use the term {\bf moat} to refer to a connected component $S$ of vertices in $\bar{F}^j$ and use $y_S^j$ to refer the dual variable associated with $S$; in order to emphasize that the moat $S$ is from a particular level $j$ and is with respect to the tight edges for that level, we will superscript the set $S$ with $j$, and denote it $S^j$.  We will increase dual variables $y_S^j$ corresponding to particular moats $S^j$.  Note that because the edges of $F$ are a subset of $\bar{F}^j$, a moat of level $j$ is a collection of the connected components of $F$.  See Figure \ref{moats} for an illustation of moats.

\begin{figure}[t]
	\centering
	\includegraphics[width=0.7\textwidth]{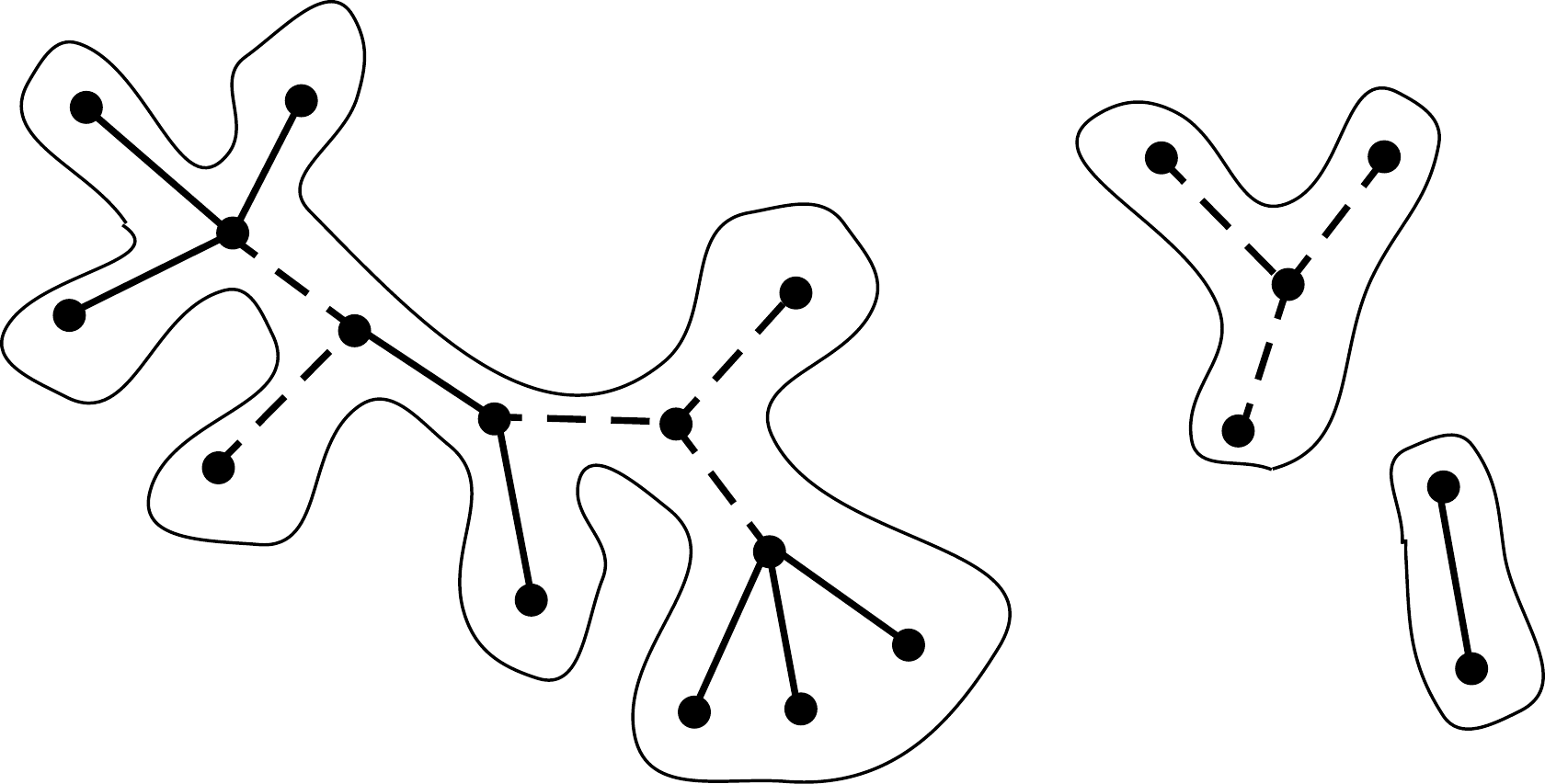}\\
	\caption{Illustration of moats at level $j$.  The solid lines represent edges in the current solution $F$, while the dashed lines represent edges that are tight in level $j$.  The moats are the connected components of the union of both the edges in $F$ and the edges tight at level $j$, and are circled.  Notice that a moat can contain multiple connected components of the current solution $F$, as the leftmost moat does.}\label{moats}
\end{figure}

A set $S \subseteq V$ is a {\bf violated} set for function $g_i$ by edges $F$ if  $|\delta(S) \cap F| < g_i(S)$; that is, if $g_i(S)=1$ but $\delta(S) \cap F=\emptyset$. Notice that for connected component $C$ of a set of edges $F$, no strict subset of $C$ can be violated.  The algorithm considers increasing duals for sets $S$ that are moats -- the connected components of $\bar{F}^j$ --  with $g_i(S)=1$, precisely because we wish to add edges to our solution from $\delta(S)$ so as to satisfy these violated sets.  We observe below that if $g_i(C)=0$ for every connected component in a set of edges $F$, then $g_i$ is satisfied by $F$, so that we can terminate the algorithm in time step $i$ when this occurs.

\begin{lemma} \label{feas}
If $g_i(C)=0$ for every connected component in a set of edges $F$, then $g_i$ is satisfied by $F$.
\end{lemma}

\begin{proof}
Note that for any set $S$, if $S$ contains some but not all of a connected component in $F$, then $|\delta(S) \cap F| \geq 1$ and so $S$ is not violated.  However, if $S$ is a union of connected components $C_j$ in $F$, then since $g_i(C_j)=0$ for each connected component $C_j$, $g_i(S) \leq \max_j g_i(C_j) = 0$, and $S$ is not violated. Thus if $g_i(C) = 0$ for all connected components $C$ of $F$, then $g_i$ is satisfied by $F$.  
\end{proof}

At the start of time step $i$, a terminal $s \in R_i$ is an {\bf active terminal} if for some connected component $X$ of the current solution $F$, we have $s\in X$ and $X$ is a violated set for function $g_i$. Let $A$ be the set of active terminals at the beginning of the time step.  Our algorithm carries out work at level $j$ then proceeds to the next level $j+1$.  If a terminal is still active when the algorithm starts its work on level $j$, we will say the terminal is {\bf active at level $j$}, and we will denote these terminals by $A_j$.  As we add edges to our solution $F$, it may be the case that for active terminal $s \in A_j$, we add edges such that $s$ is in a connected component $X$ of $F$ with $g_i(X) = 0$; at this point $s$ is no longer active. We may say that $s$ has become {\bf inactive}; it was {\bf previously active}. We denote the set of all terminals that were previously active at level $j$ (at any time step) as $P_j$.   Also, as we increase dual variables, a terminal $s$ active at level $j$ may reach its limit at level $j$; that is, $\sum_{S: s \in S} y_S^j = 2^j$.  In this case, we move $s$ from $A_j$ to $P_j$.

A moat $S^j$ is an {\bf active moat} if $g_i(S^j)=1$ and its corresponding dual variable $y_S^j$ has not yet reached its limit in level $j$.  Note that an active moat $S^j$ is a violated set for $g_i$ by edges $\bar{F}^j$ since $g_i(S^j)=1$ and $\delta(S^j) \cap \bar F^j = \emptyset$ because moat $S^j$ is a connected component of $\bar{F}^j$.  We denote the current set of active moats by $\mathcal{M}$. We say a dual variable $y_S^j$ is an {\bf active dual variable} if its corresponding moat $S^j$ is active.

\section{The Algorithm and Its Analysis}
\label{sec:alg}
\subsection{The Primal-Dual Online Algorithm}

Our algorithm (see Fig. \ref{alg}) is a dual ascent algorithm in which we grow active dual variables, starting at lowest level $j$.   We increase dual variables around active terminals in level $j$ and buy paths between terminals until either all terminals are inactive, or we can no longer increase dual variables around active terminals, since the dual variables have reached their limits.  Then we proceed to level $j+1$.

More precisely, our algorithm starts with $F = \emptyset$ and $y_S^j = 0$ for all $j$ and all $S\subseteq V$. At the beginning of each time step $i$, the function $f_i$ arrives and some non-terminal nodes in $V$ may become terminals. We  update active terminal set $A$ and active moat set $\mathcal{M}$. Conceptually we loop through the levels $j$, starting at level $-\infty$ and continuing to level $\infty$; we explain below how we can omit very small and very large values of $j$ so that the algorithm is implementable in polynomial time.  For each level $j$, we execute two distinct while loops; we call the first the {\bf consolidate} loop for level $j$, and the second the {\bf dual growth} loop for level $j$. In the consolidate loop, we add edges in $F-\bar{F}^j$ to $\bar{F}^j$ one at a time; adding such an edge may cause two moats to be merged.  We then add paths to $F$ connecting any pair of terminals $s_1 \in A_j$ ($s_1$ currently active) and $s_2 \in P_j$ ($s_2$ previously active) in the same moat that were not already connected in $F$.   In  the dual growth loop, while there are still active terminals at level $j$,  our algorithm uniformly increases all active dual variables $y_S^j$ until: (1) an active $y_S^j$ reaches its limit in level $j$; (2) an edge $e\in E$ becomes tight in level $j$; we then add $e$ to $\bar{F}^{j}$; (3) two terminals $s_1 \in A_j$ and $s_2 \in A_j \cup P_j$ {\bf connect in level $j$}; that is, there is a path of edges between them  that are either tight or in $F$.  We then let $p$ be this path of edges (that are either tight or in $F$) connecting $s_1$ and $s_2$ that minimizes $\sum_{e \in p-F} c_e$; we build path $p$ in $F$, and update the set $A$ of active terminals and the set  $\mathcal{M}$ of active moats. We output $F$ as the solution for $(IP_i)$.

We remark that the consolidate loop serves two purposes. First, by adding $F$ to $\bar{F}^j$, it ensures that each component of $\bar{F}^j$ is a collection of the connected components of $F$. Second, it ensures that the level-$j$ terminals that are contained in the same level-$j$ moat $S^j$ are contained in the same connected component of $F$ (Lemma \ref{unique}). These properties imply that active terminals are always contained in an active moat and thus the algorithm is well-defined.

The following example illustrates the algorithm and the necessity of the consolidate loop. Consider Figure \ref{fig:alg-ex}: the input graph consists of 4 terminals $s_1, s_2, s_3, s_4$ on a line and the proper function $g$ is such that $g(S) = |S| \mod 2$. For levels $j < -1$, all terminals are active and the algorithm grows dual variables around each of them, but the dual variables reach their limit without any edges going tight. Thus, the algorithm starts level $-1$ with $F = \emptyset$ and all terminals still active. At the end of the level, the edge $(s_2,s_3)$ goes tight and gets added to $F$. The terminals $s_2$ and $s_3$ then become inactive. At the beginning of level $0$, the consolidate loop adds $(s_2,s_3)$ to $\bar{F}^0$. The algorithm then grows dual variables around each of the remaining active terminals $s_1$ and $s_4$. However, these dual variables reach their limit before any edge goes tight. At the beginning of level $1$, the consolidate loop again adds the edge $(s_2,s_3)$ to $\bar{F}^{1}$. At the end of the level, the edges $(s_1, s_2)$ and $(s_3, s_4)$ goes tight and are added to $\bar{F}^{1}$. At this point, $\bar{F}^{1}$ contains a path connecting the remaining two active terminals $s_1$ and $s_4$, so the edges $(s_1,s_2)$ and $(s_3,s_4)$ are added to $F$. There are no remaining active terminals and $F$ is a feasible solution. Now, we argue that the algorithm is not well-defined without the consolidate loop. Consider the algorithm without the consolidate loop. The algorithm essentially behaves in the same way for levels below level $1$. Now,  the algorithm starts level $1$ with $\bar{F}^1 = \emptyset$. The edges $(s_1,s_2)$ and $(s_3,s_4)$ still go tight during the dual growth loop, but note that once they get added to $\bar{F}^1$, the level-$1$ moats are $S_1 = \{s_1,s_2\}$ and $S_2 = \{s_3,s_4\}$. Since $g(S_1) = g(S_2) = 0$, there are no more moats even though $s_1$ and $s_4$ are still active. Thus, the consolidate loop is necessary for the algorithm to be well-defined.

\begin{figure}[ht]
\centering
\subfigure[Input graph with 4 terminals.]{
  \includegraphics[scale=0.3]{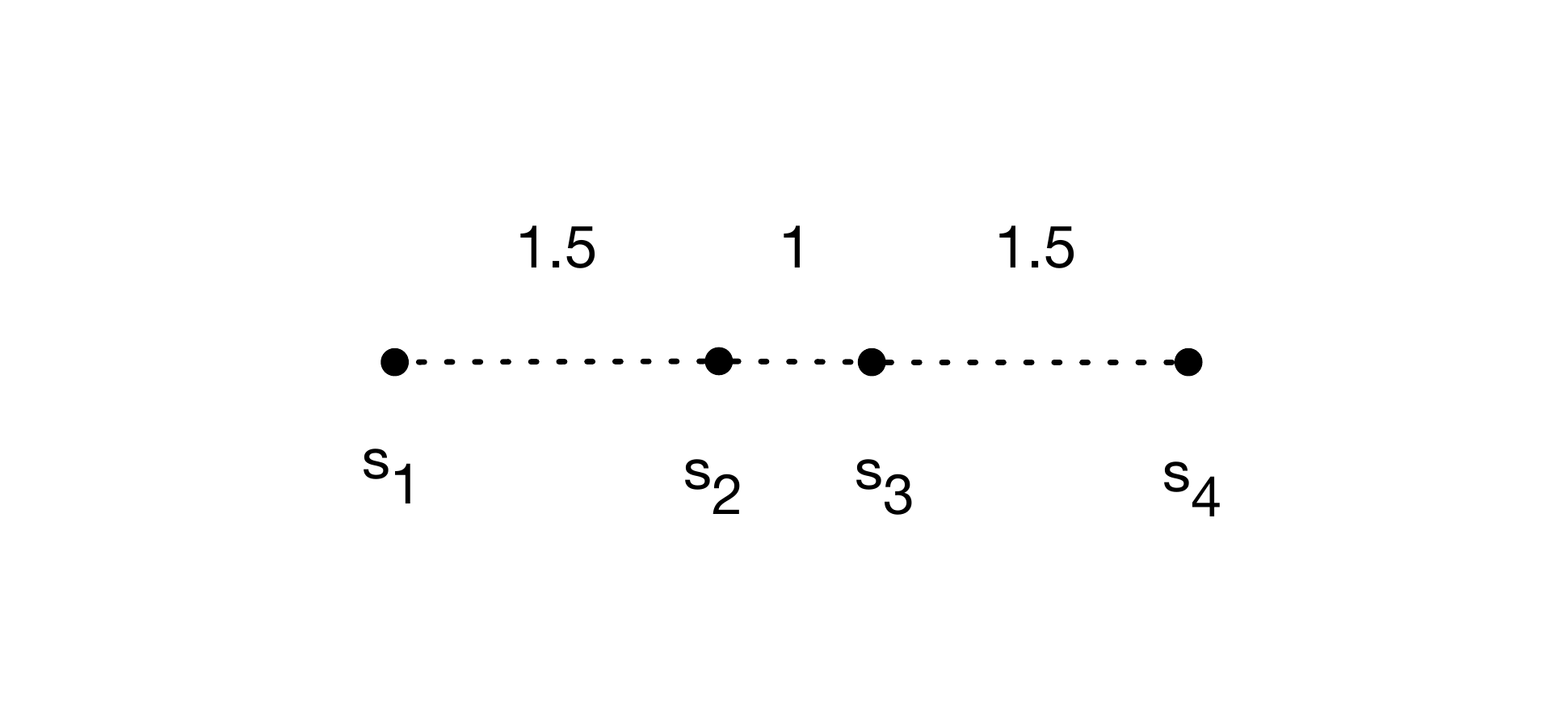}
  \label{fig:alg-ex1}}
\quad
\subfigure[Level -1.]{
  \includegraphics[scale=0.3]{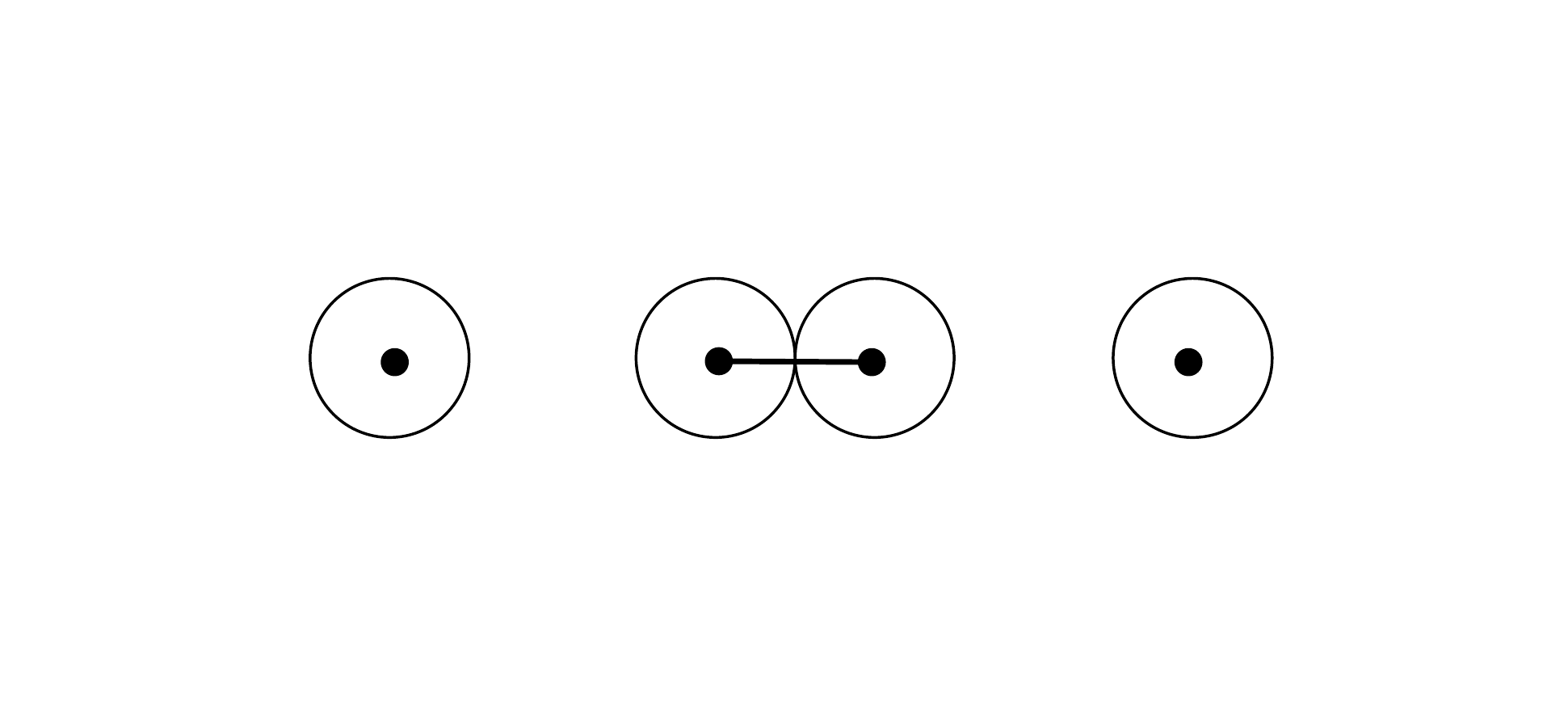}
  \label{fig:alg-ex2}}
\subfigure[Level 0.]{
  \includegraphics[scale=0.3]{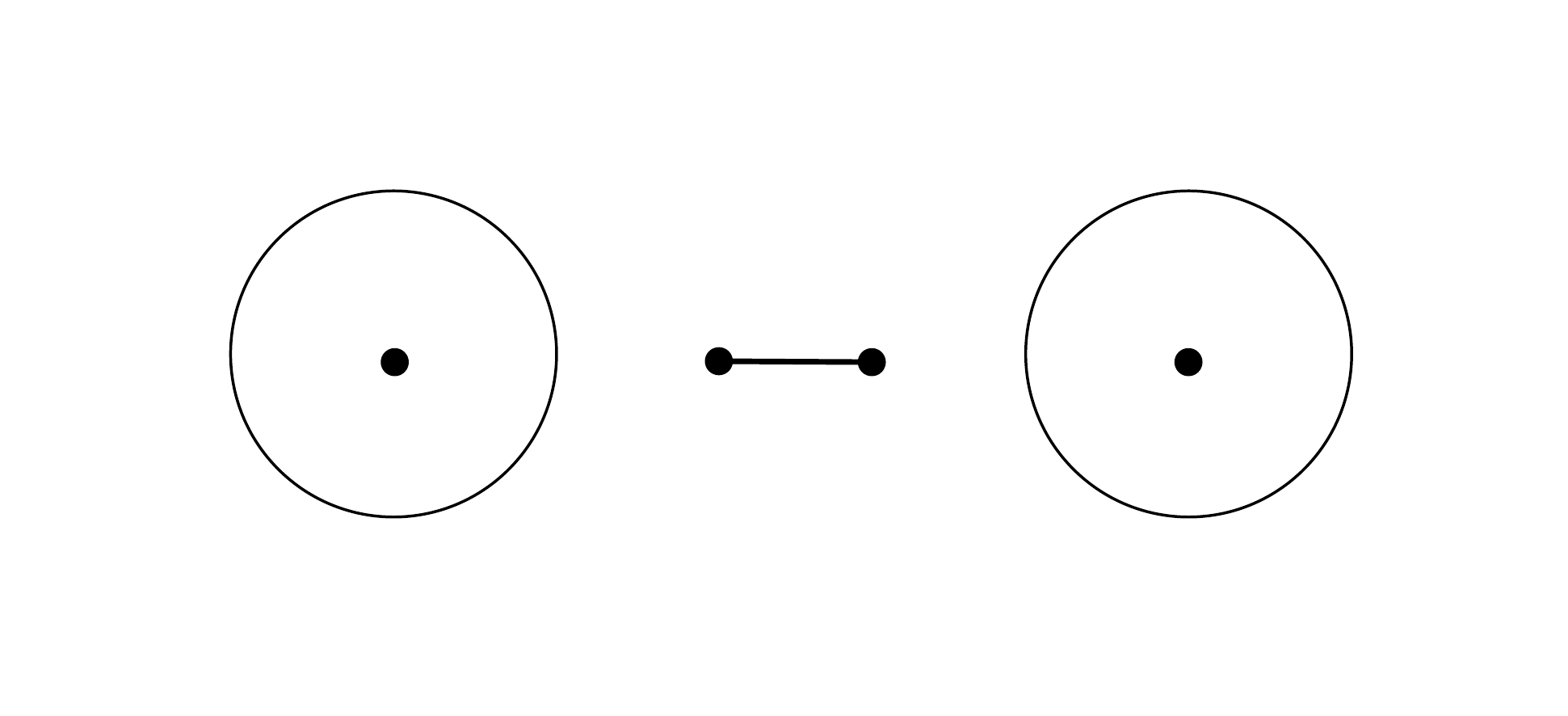}
  \label{fig:alg-ex3}}
\subfigure[Level 1.]{
  \includegraphics[scale=0.3]{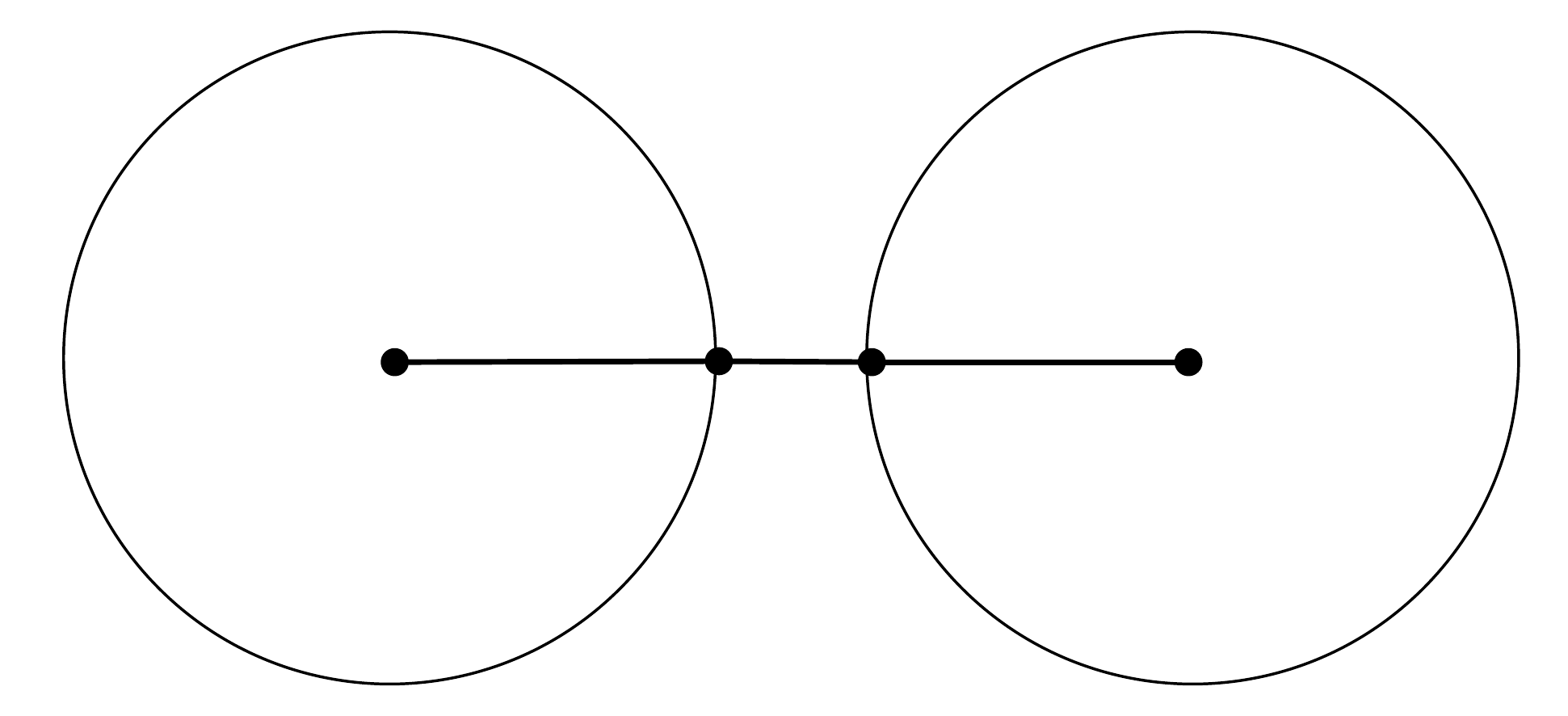}
  \label{fig:alg-ex4}}
  \caption{An example run of the algorithm on a graph with 4 terminals and proper function $g$ such that $g(S) = |S| \mod 2$. The solid edges represent the edges in $F$ and the circles represent the moats at the end of each level.}
  \label{fig:alg-ex}
\end{figure}

\begin{figure}[t]
\begin{center}
{
\algorithm{Algorithm}
{
$F = \emptyset$, $\bar{F}^{j} = \emptyset$ for all $j$, and $y_S^{j} = 0$ for all $j$ and $S\subseteq V$\\
For each $\{0,1\}$-proper function $f_i$ that arrives\\
\> Update active terminals $A$, and active moats $\mathcal{M}$\\
\> For $j \gets -\infty$ to $\infty$\\
\> \> (Consolidate) While there is an edge $\bar e \in F \setminus \bar F^j$\\
\> \> \> Add $\bar e$ to $\bar F^j$\\
\> \> \> While there are terminals $s_1 \in A_j$, $s_2 \in P_j$ in the same moat $S^j$\\
\> \> \> that are not connected in $F$\\
\> \> \> \> Let $p \subseteq E$ be an $s_1$-$s_2$ path in $\bar F^j$ minimizing $\sum_{e \in p-F} c_e$\\
\> \> \> \> $F \gets F \cup \{p\}$, i.e. build edges $p-F$\\
\> \> \> \> Update $A$\\
\> \> Update active moats $\mathcal{M}$\\
 \> \> (Dual growth) While there are terminals active at level $j$\\
 \> \>\> Grow uniformly all active dual variables $y_S^{j}$ until\\
 \> \>\> 1) An active $y_S^j$ reaches its limit in level $j$\\
 \> \>\> 2) An edge $e\in E$ becomes tight in level $j$, then $\bar{F}^{j} = \bar{F}^{j} \cup \{e\}$ \\
\>  \>\> 3) Two terminals $s_1 \in A_j$ and $s_2 \in A_j \cup P_j$ connect in level $j$, then\\
\>  \>\>\>      Let $p \subseteq E$ be the $s_1$-$s_2$ path of edges in $\bar{F}^j$ minimizing $\sum_{e \in p-F} c_e$\\
\>  \>\>\>      $F = F \cup \{p \}$, i.e. build edges $p-F$ \\
\>  \>\>\>     Update $A$\\
  \>\>\> Update active moats $\mathcal{M}$
}}
\end{center}
\caption{Primal-dual algorithm for the online proper constrained forest problem}
\label{alg}
\end{figure}

The algorithm in Figure \ref{alg} can be implemented in polynomial time.  We assume that all edge costs $c_e$ are integers.  Then as a matter of algorithmic implementation, we do not need to maintain levels $j < -1$ or start the loop for $j < -1$, since for such levels dual variables will reach their limits before any edge $e$ can go tight. We show below (in Theorem \ref{termination}) that we do not need to maintain levels $j > \lceil \log_2 (\max_{u,v \in V} d(u,v)) \rceil$ or continue the loop for such values of $j$, where $d(u,v)$ is the distance in $G$ between $u$ and $v$ using edge costs $c_e$; intuitively, we will have generated a feasible solution $F$ in the levels below this one since the dual variables will not reach their limit before all edges in each possible $u$-$v$ shortest path are tight and all terminals will connect.  Thus we need only maintain $O(\log (\max_{u,v \in V} d(u,v)))$ different levels and dual solutions $y^j$, which is polynomial in the input size. Finding the active moats involves computing connected components in the set of tight edges $\bar F^j$ and checking whether each component is a violated set.  In each iteration, we can iterate through all the edges and active dual variables for the current level, of which there are at most a polynomial number, to see which of conditions (1)-(3) will hold first given a uniform increase of the active dual variables.  Since there are at most $|R_i|$ active terminals, and each iteration either reduces the number of active dual variables, makes an additional edge tight, or merges two disjoint moats, there can be at most a polynomial number of iterations for each level.  Since there are at most a polynomial number of levels to consider, the entire algorithm will take polynomial time.

\subsection{The Analysis}

We will now state our main theorem.

\begin{theorem} The algorithm of Figure \ref{alg} is an $O(\log |R_i|)$-competitive algorithm for the online proper constrained forest problem $(IP_i)$.
\label{thm}
\end{theorem}

We begin with a summary of what will follow.   We show in Lemma \ref{lem1} that in each time step, the solution $F$ is a feasible primal solution to the integer program, and each $y^j$ is a feasible dual solution.  As mentioned at the end of the introduction, the basic argument is a charging scheme in which we charge the cost of the edges in $F$ to the dual variables, in such a way that the cost of all the edges is at most the sum of the dual variables $y^j$ summed over all levels $j$.  We will in Lemma \ref{dualvars} show that because the dual growth for each level  $j$ is limited by $2^j$, only the top $O(\log |R_i|)$ levels account for almost all the total dual value; levels below the top $O(\log |R_i|)$ have a negligible amount of dual value.  Recall that the dual solution $y^j$ for each level $j$ is a lower bound on the cost of an optimal solution.  Thus since the cost of the edges in $F$ is essentially at most the value of the dual solutions of the top $O(\log |R_i|)$ levels, and each one is a lower bound on the cost of an optimal solution, the cost of the edges in $F$ are at most a factor of $O(\log |R_i|)$ from the cost of an optimal solution.

In order to perform the charging scheme, we will show in Lemma \ref{moat} that the growth of a dual variable $y_S^j$ can be uniquely credited to some connected component $X$ of the set of edges $F$.  The charging scheme will maintain accounts for all the current connected components of the set of edges $F$.  The key part of the analysis is Lemma \ref{cost}, which shows that  at any point in the algorithm, the total sum of the dual variables $y^j$ summed over all levels $j$ is equal to the cost of the edges currently in $F$ plus the credits in the accounts summed over all the components $X$ of $F$; these accounts will let us pay for adding edges to $F$ in the future.

The proofs below are based on, but substantial generalizations of, those given in Berman and Coulston \cite{BermanC97}.

We can now start the main analysis of the algorithm.  The following lemma is key to both the charging scheme and to proving the termination of the algorithm.  See Figure \ref{moats2} for an illustration.

\begin{lemma} \label{unique} In every iteration of the dual growth loop at level $j$, for each moat $S^j$, the subset of $A_j \cup P_j$ contained in $S^j$ is contained in a unique connected component $X$ of $F$.
\label{moat}
\end{lemma}

\begin{proof} 
The proof follows directly from the algorithm. The consolidate loop and Step (3) of the dual growth loop ensures that whenever $s_1, s_2 \in A_j \cup P_j$ are connected in $\bar{F}^j$ during the dual growth loop, then they are connected in $F$ as well. Since each moat $S^j$ is a connected component of $\bar{F}^j$, the terminals of $A_j \cup P_j$ contained in the moat $S^j$ are contained in a unique connected component $X$ of $F$.  
\end{proof}

\begin{figure}[t]
	\centering
	\includegraphics[width=0.5\textwidth]{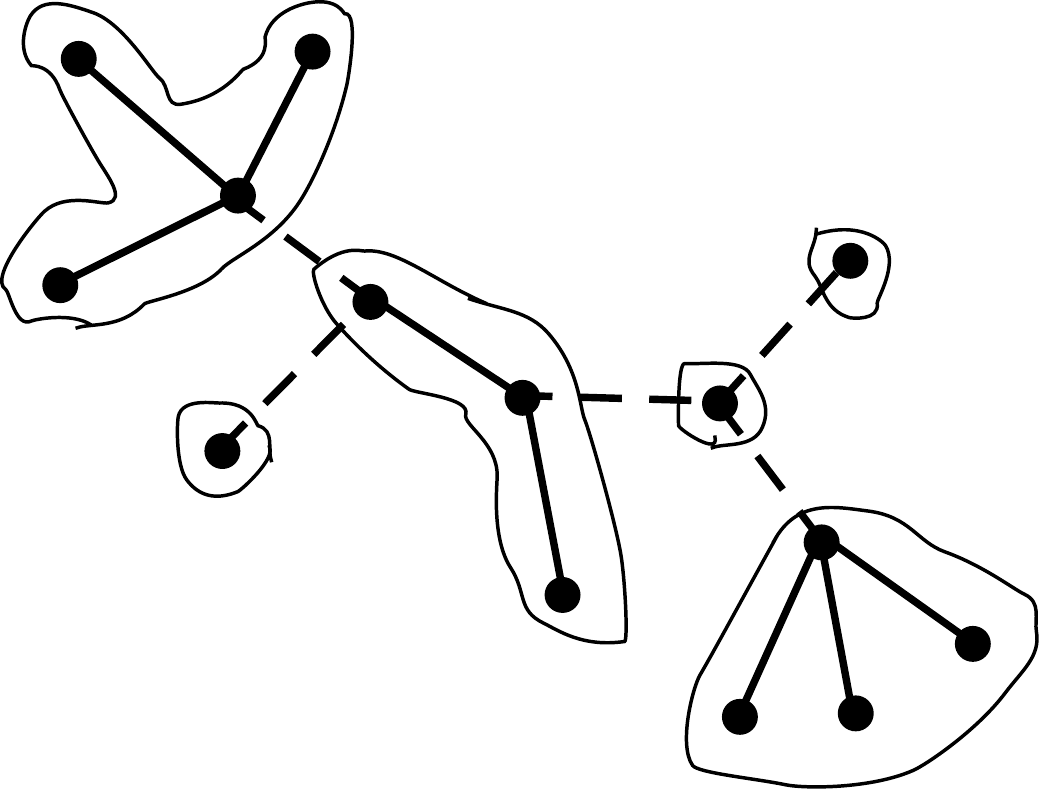}\\
	\caption{A moat $S$ with the connected components of $F$ circled.  Lemma \ref{unique} asserts that in the dual growth loop, all active and previously inactive vertices in the moat will be in exactly one of these components.  Lemma \ref{active} observes that because this is the case, whether the moat $S$ is active depends exactly on whether any terminal in this one component is active.}\label{moats2}
\end{figure}

We now turn to showing that the algorithm is well-defined and that it terminates.  We need the following lemma to begin.

\begin{lemma} \label{inactive}
At any time during the execution of the algorithm, if a connected component $X$ of $F$ has no active terminal in it, then $g_i(X) = 0$.
\end{lemma}

\begin{proof}
  There are two cases to consider: (1) $X$ is a singleton set; (2) $X$ was formed by adding a $s_1$-$s_2$ path $p$ to $F$ that connects several smaller components. The statement clearly holds for case (1). Let us now consider case (2). By definition of the algorithm, at least one of $s_1$ or $s_2$ was an active terminal before $p$ was added to $F$. Suppose $s_1$ was the active terminal. After $p$ was added to $F$, $s_1$ is contained in the new component $X$ but is no longer active. By definition of the algorithm, this can only happen if $g_i(X) = 0$. Thus, the statement holds in case (2) as well.
 \end{proof}

The following lemma shows that the algorithm is well-defined.

\begin{lemma} \label{active}
In every iteration of the dual growth loop at level $j$, if a terminal is active, then it is contained in a moat $S^j$ that is active. In particular, the dual variable $y_S^j$ is active.
\end{lemma}

\begin{proof}
Let $s$ be an active terminal and $S^j$ be the moat containing $s$.  Observe that $S^j$ is a union of some connected components of $F$; this is because the algorithm ensures that $\bar{F}^j$ contains $F$ and moats are connected components of $\bar{F}^j$. By Lemma \ref{moat}, there is a unique connected component $X$ of $F$ in $S^j$  that contains all the active terminals in $S^j$. Because $s$ is active and is contained in $X$, we have $g_i(X) = 1$ (since otherwise $s$ would become inactive).  For any other connected component $X'$ of $F$ contained in $S^j$, $X'$ does not contain any active terminal so $g_i(X') = 0$ by Lemma \ref{inactive}. Let $Z$ be the union of these connected components $X'$; then by the definition of proper functions it must be that $g_i(Z) \leq \max_{X'} g_i(X') = 0$. So we have that $S^j$ is partitioned into sets $X$ and $Z$. Because $g_i$ is proper, $g_i(S^j)=g_i(V-S^j)$, and $V-S^j$ and $Z$ partition $V - X$, so that $g_i(X) = g_i(V-X) \leq \max(g_i(V-S^j), g_i(Z)) = \max(g_i(S^j),g_i(Z)).$  Thus $g_i(S^j)=0$ would imply $g_i(X)=0$, a contradiction.  Thus $g_i(S^j)=1$ and $S^j$ is an active moat. 
\end{proof}

Finally, we can prove that the algorithm terminates and that it does not use any level beyond level $\lceil \log_2 (\max_{u,v \in V} d(u,v)) \rceil$.

\begin{theorem}  The algorithm terminates in each time step $i$, and will find a feasible solution before it reaches a level greater than $\lceil \log_2 (\max_{u,v \in V} d(u,v)) \rceil$. \label{termination}
\end{theorem}

\begin{proof}
First, we argue that the consolidate and dual growth loops at each level must terminate. We observe that each iteration through the consolidate loop at level $j$ joins two components of $F$; once we have merged $n$ components, we have a tree spanning all vertices, which is a feasible solution to the problem, so there can be at most $n$ iterations of the consolidate loop. Each iteration through the dual growth loop either adds a tight edge to a level, joins two components of $F$, or causes a terminal to reach its limit on the level; thus there can be at most $m + n + |R_i|$ iterations through the dual growth loop at any level.

Next, we show that the algorithm terminates by level $j = \lceil \log_2 (\max_{u,v \in V} d(u,v)) \rceil$. Suppose, towards a contradiction, that this is not the case. Then, by Lemmas \ref{feas} and \ref{inactive}, there are some active terminals $A$ at the end of level $j$. These terminals cannot reach their limit until all the edges of all shortest $u$-$v$ paths are tight (for all pairs $u,v \in A$). Thus, $A$ is connected in $\bar{F}^j$, and so by Lemma \ref{unique}, there is a connected component $X$ of $F$ containing $A$. By Lemma \ref{inactive}, every other connected component $X'$ of $F$ has $g_i(X') = 0$ since it does not contain any active terminal. Since the union of these components is $V \setminus X$, by the definition of proper functions, we have $g_i(V \setminus X) = 0$. But then $g_i(X) = g_i(V \setminus X)$ as well, so every connected component $C$ of $F$ has $g_i(C) = 0$, and thus $F$ is feasible by Lemma \ref{feas}. Therefore, by level $\lceil \log_2 (\max_{u,v \in V} d(u,v)) \rceil$ we have found a feasible solution for $(P_i)$, and step $i$ must terminate.
 \end{proof}

\begin{theorem} \label{lem1} At the end of time step $i$ of the algorithm in Figure \ref{alg}, $F$ is a feasible solution to $(IP_i)$ and each dual vector $y^{j}$ is a feasible solution to $(D_i)$.
\end{theorem}

\begin{proof} Our algorithm terminates each time step $i$ when there are no active terminals, and thus by Lemma \ref{inactive}, for each connected component $X$ of $F$, $g_i(X) = 0$.  Thus by Lemma \ref{feas}, the solution is feasible for $(P_i)$.  By construction of the algorithm each dual solution $y^j$ is feasible for $(D_i)$ since we stop growing a dual $y^j_S$ if it would violate a dual constraint.
 \end{proof}

We now turn to analyzing the cost of the solution returned by the algorithm.  As discussed previously, in order to give a bound on the total cost of edges in $F$, we create an account for each connected component $X$ in $F$, denoted $\mathrm{Account}(X)$. We will define a shadow algorithm to credit potential to accounts as duals are increased and remove potential from accounts to pay for building edges. We will show that the total cost of edges in $F$ plus the total unused potential remaining in all accounts is always equal to the sum of all dual variables over all levels, i.e. $\sum_j\sum_S y_S^{j}$.

Our shadow algorithm works as follows.  First, whenever we increase an active dual variable $y_S^{j}$, we will credit the amount of increase to $\mathrm{Account}(X)$, where $X$ is the unique connected component in $F$ that contains all terminals in the moat $S^j$ that are in $P_j \cup A_j$, as given by Lemma \ref{moat}. Second, whenever the algorithm builds a path $p$ in $F$ connecting two terminals $s_1$ and $s_2$ from $P_j \cup A_j$, we let $X_k$ be the resulting connected component in $F$ that contains $s_k$ for $k=1,2$. As a result of building edges $p-F$, $X_3 = X_1 \cup X_2 \cup \{p-F\}$ will become a connected component in $F$. We will merge unused potential remaining in $\mathrm{Account}(X_1)$  and $\mathrm{Account}(X_2)$ into $\mathrm{Account}(X_3)$ and remove potential from  $\mathrm{Account}(X_3)$ to pay for the cost of building edges in $p-F$. A key part of our analysis is to bound the cost of $p - F$ against the dual growth of terminals in $X_1$ and $X_2$. This will then let us show that the account of the ``smaller'' component can pay for $p - F$.

We will need the following helper lemmas to prove our desired statements about the accounts. Define $\mathrm{Growth}(X,j)$ to be the maximum total dual growth of a terminal in $X$ in level $j$; so $$\mathrm{Growth}(X, j) = \max_{s\in X} \{\sum_{S \subseteq V: s \in S} y_S^{j}\mbox{ and } s \in A_j \cup P_j \}.$$  Observe that $\mathrm{Growth}(X,j) \leq 2^j$ by the limit on dual growth on level $j$. For example, consider the instance given in Figure \ref{fig:alg-ex}. At the end of level 0, the components of $F$ are $\{s_1\}$, $\{s_2,s_3\}$ and $\{s_4\}$. We have $\mathrm{Growth}(\{s_2,s_3\},0) = 0$ (since neither $s_2$ nor $s_3$ is active in this level) and $\mathrm{Growth}(\{s_1\},0) = \mathrm{Growth}(\{t_1\},0) = 2^{0}$. 

We now work towards proving that whenever we buy a path $p$ connecting two components $X_1$ and $X_2$ during level $j$, the level $j$ dual growth of both components can pay for the set of new edges $p-F$, i.e. $\mathrm{Growth}(X_1, j) + \mathrm{Growth}(X_2, j) \geq \sum_{e \in p-F} c_e$.  Note that edges in $\bar F^j$ are added one by one in both the consolidate and dual growth loops. Define the \textbf{first moat} of level $j$ that contains $v$ to be the connected component of $\bar F^j$ containing $v$ at the first time that $v$ is connected to a terminal of $A_j \cup P_j$ in $\bar F^j$. The following crucial lemma allows us to charge the cost of buying paths in $\bar F^j$ to the dual growth of a single terminal.

\begin{lemma}
  \label{path-cost}
  Let $v \in V$ be a vertex and $S_v$ be the first moat of level $j$ that contains $v$. There exists a terminal $s \in A_j \cup P_j$ in $S_v$ with a $s$-$v$ path $p$ in $\bar F^j$ with cost
\[\sum_{e \in p-F} c_e \leq \sum_{S \subset S_v: s \in S} y^j_S \leq \mathrm{Growth}(X,j),\]
where $X$ is the component in $F$ containing $s$.
\end{lemma}

\begin{proof}
  We prove this by induction on $\bar F^j$, as edges are added to $\bar F^j$. The statement clearly holds in the beginning, when $\bar F^j = \emptyset$. Now, we turn to the inductive case. Suppose $v$ was first connected to a terminal of $A_j \cup P_j$ when the edge $\bar e$ was added to $\bar F^j$, and let $t$ be that terminal. Suppose $S_1$ was the moat containing $t$ before $\bar e$ was added. Since $v$ was not connected to any terminal of $A_j \cup P_j$ before this time, we have that $v$ is an endpoint of $\bar e$, and the other endpoint of $\bar e$, say $u$, is contained in $S_1$. In particular, the moat $S_v = S_1 \cup \{v\}$ is the first moat containing $v$. 

Since $u$ was connected to a terminal of $A_j \cup P_j$ (in particular, the terminal $t$) in $\bar F^j$ at an earlier time, the inductive hypothesis implies that there exists a terminal $s \in A_j \cup P_j$ and a $s$-$u$ path $q$ in $\bar F^j$ with cost 
\[\sum_{e \in q-F} c_e \leq \sum_{S \subset S_u : s \in S}y^j_S,\]
where $S_u \subseteq S_1$ is the first moat containing $u$, and it also contains $s$. See Figure \ref{moat-path} for an illustration.

\begin{figure}[t]
	\centering
	\includegraphics[width=0.3\textwidth]{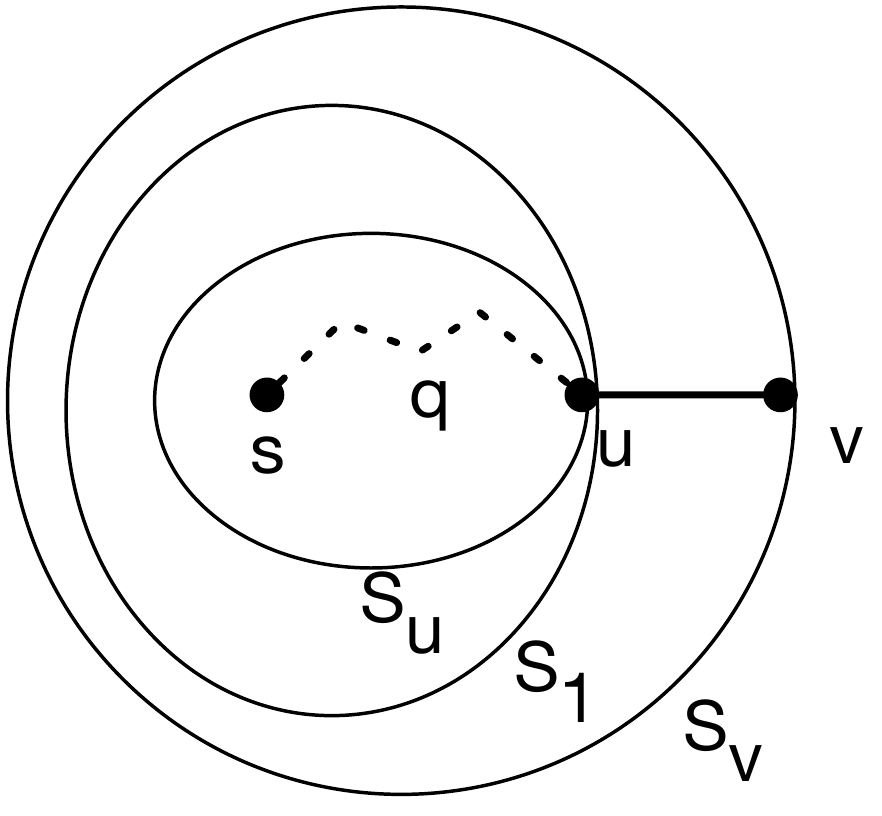}\\
	\caption{An illustration of the proof of Lemma \ref{path-cost}.}\label{moat-path}
\end{figure}

Now consider the $s$-$v$ path $p = q \cup \{\bar e\}$. The edge $\bar e$ was either an edge of $F$ added to $\bar F^j$ during the consolidate loop or it was an edge that went tight. In the first case, we are done. In the second case, we have 
\[c_{\bar e} = \sum_{S \subseteq V: (u,v) \in \delta(S)}y^j_S = \sum_{S \subseteq S_1 : u \in S}y^j_S = \sum_{S_u \subseteq S \subseteq S_1}y^j_S,\]
where the first equality follows from tightness of $\bar e$, the second from the fact that $v$ was never contained in an active moat before now, and the third from the fact that $S_u$ was the first moat containing $u$. Therefore, the cost of $p$ is
\begin{align*}
  \sum_{e \in p-F} c_e 
  &= \sum_{e \in q - F} c_e + c_{\bar e} \\
  &\leq \sum_{S \subset S_u : s \in S} y^j_S + \sum_{S_u \subseteq S \subseteq S_1} y^j_S \\
  &\leq \sum_{S \subseteq S_1 : s \in S} y^j_S\\
  &\leq \sum_{S \subset S_v : s \in S}y^j_S,
\end{align*}
where the second inequality follows from the fact that $S_u$ contains $s$ and the last inequality follows from the fact that $S_1$ is a strict subset of $S_v$. 
\end{proof}

\begin{lemma}
  \label{merge-cost}
  Suppose there are terminals $s_1, s_2 \in A_j \cup P_j$ in different components of $F$ ($X_1$ and $X_2$ respectively) such that there is a path between $s_1$ and $s_2$ in $\bar F^j$. Let $p$ be a path in $\bar F^j$ that minimizes the cost $\sum_{e \in p - F} c_e$. Then $\sum_{e \in p - F}c_e \leq \mathrm{Growth}(X_1,j) + \mathrm{Growth}(X_2,j)$.
\end{lemma}

\begin{proof}
  To prove the lemma, we will show that there exists a path $p$ with cost $\sum_{e \in p - F}c_e \leq \mathrm{Growth}(X_1,j) + \mathrm{Growth}(X_2,j)$. We will consider the consolidate loop and the dual growth loops separately. In the consolidate loop, there are two cases: (1) either $s_1$ and $s_2$ were connected in $\bar F^j$ even before any edge of $F \setminus \bar F^j$ was added to $\bar F^j$; (2) or $s_1$ and $s_2$ were only connected in $\bar F^j$ after some edge of $F \setminus \bar F^j$ was added to $\bar F^j$. Suppose $s_1 \in A_j$ and $s_2 \in P_j$. Case 1 can occur if $s_1$ was already contained in the moat containing $s_2$ in some previous time step (when $s_1$ was not yet a terminal). This case is easy: Lemma \ref{path-cost} implies that there exists a path $p'$ in $\bar F^j$ connecting $X_2$ and $s_1$ with cost $\sum_{e \in p' - F} c_e \leq \mathrm{Growth}(X_2,j)$. Next, we consider Case 2. Let $\bar e = (u,v)$ be the edge of $F$ added to $\bar F^j$ that caused $s_1$ and $s_2$ to connect in $\bar F^j$. Suppose that before $\bar e$ was added, $s_1$ was connected to $u$, and $s_2$ was connected to $v$ in $\bar F^j$. Applying Lemma \ref{path-cost} to the components $X_1$ and $X_2$, we get that there is a path in $p_1$ in $\bar F^j$ and a path $p_2$ in $\bar F^j$ with cost $\sum_{e \in p_1 - F} c_e \leq \mathrm{Growth}(X_1,j)$ and $\sum_{e \in p_2 - F} c_e \leq \mathrm{Growth}(X_2,j)$. Since $\bar e \in F$, the path $p_1$ followed by the edge $e$ followed by the path $p_2$ is a path in $\bar F^j$ with cost at most $\mathrm{Growth}(X_1, j) + \mathrm{Growth}(X_2,j)$.

  Finally, we consider the dual growth loop. This case is similar to the second case of the consolidate loop, but we also need to show that the dual growth can also pay for the edge $\bar e$. Suppose $S_1$ and $S_2$ are the moats containing $s_1$ and $s_2$ before $\bar e$ was added, and $u \in S_1$ and $v \in S_2$. (See Figure \ref{colliding-moats} for an illustration.) Since $\bar e$ is a tight edge, we have
\[c_{\bar e} = \sum_{S \subseteq V: (u,v) \in \delta(S)}y^j_S = \sum_{S \subseteq S_1: u \in S}y^j_S + \sum_{S \subseteq S_2 : v \in S}y^j_S.\]
Let $p$ be the path in $\bar F^j$ that minimizes the cost $\sum_{e \in p-F} c_e$. Since $s_1$ and $s_2$ were only connected after $\bar e$ was added, the path $p$ contains the edge $\bar e$. Let $p_1$ be the subpath of $p$ from $s_1$ to $u$ and $p_2$ be the subpath of $p$ from $v$ to $s_2$. The cost of $p$ is
\[\sum_{e \in p_1 - F} c_e + c_{\bar e} + \sum_{e \in p_2-F}c_e = \left(\sum_{e \in p_1 - F} c_e + \sum_{S \subseteq S_1 : u \in S} y^j_S\right) + \left(\sum_{S \subseteq S_2 : v \in S}y^j_S + \sum_{e \in p_2-F}c_e\right).\]

\begin{figure}[t]
	\centering
	\includegraphics[width=0.4\textwidth]{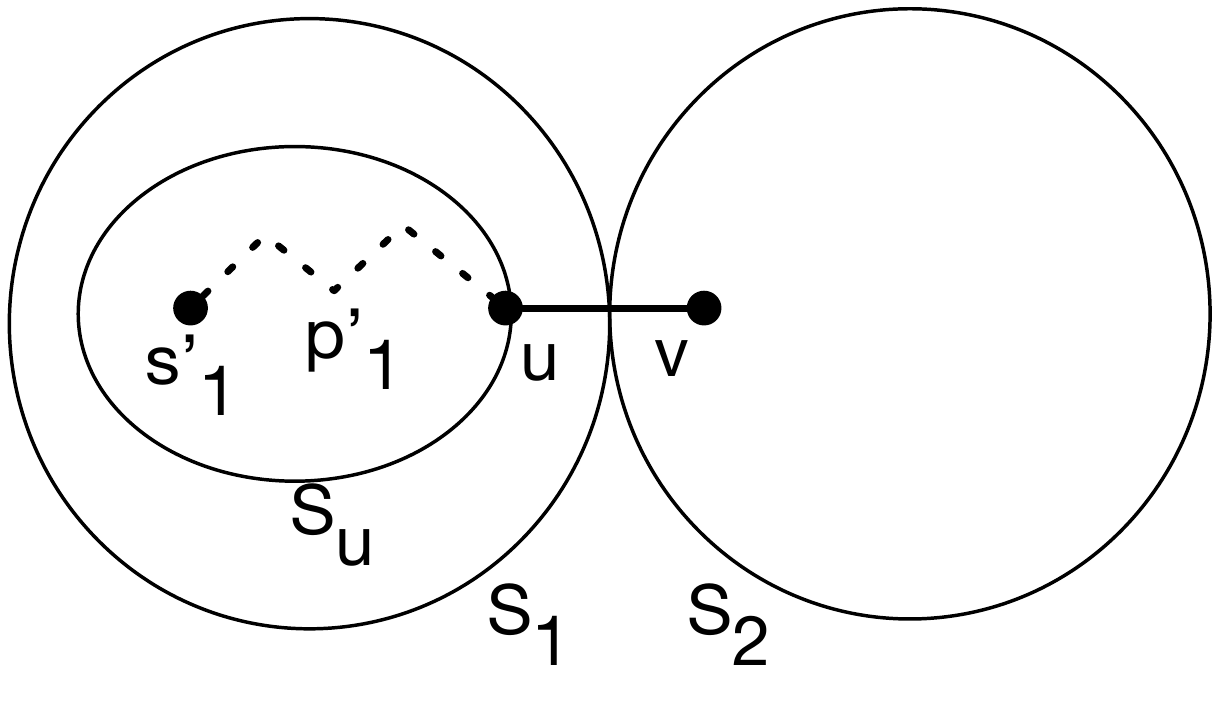}\\
	\caption{An illustration of the proof of Lemma \ref{merge-cost}. The total dual growth around $s'_1$ can pay for the cost of $p'_1$ as well as the portion of $(u,v)$ ``contained'' in $S_1$.}\label{colliding-moats}
\end{figure}

We claim that 
\[\sum_{e \in p_1 - F} c_e + \sum_{S \subseteq S_1 : u \in S} y^j_S \leq \mathrm{Growth}(X_1,j).\] Let $S_u \subseteq S_1$ be the first moat containing $u$. Lemma \ref{path-cost} implies that there exists a terminal $s'_1 \in S_u$ and a $s'_1$-$u$ path $p'_1$ with cost $\sum_{e \in p'_1 - F} c_e \leq \sum_{S \subset S_u : s'_1 \in S}y^j_S$. By Lemma \ref{unique}, $s'_1$ and $s_1$ are already connected in $F$, and so the cost of $p_1$ is at most the cost of $p'_1$. 
Thus, we have
\begin{align*}
  \sum_{e \in p_1 - F} c_e + \sum_{S \subseteq S_1 : u \in S} y^j_S 
  &\leq \sum_{S \subset S_u : s'_1 \in S}y^j_S + \sum_{S \subseteq S_1 : u \in S} y^j_S \\
  &\leq \sum_{S \subset S_u : s'_1 \in S}y^j_S + \sum_{S_u \subseteq S \subseteq S_1} y^j_S \\
  &\leq \sum_{S \subseteq S_1 : s'_1 \in S} y^j_S\\
  &\leq \mathrm{Growth}(X_1,j),
\end{align*}
where the second inequality follows from the fact that $S_u$ is the first moat containing $u$, the third from the fact that $S_u$ contains $s'_1$, and the final from the definition of $\mathrm{Growth}$. This proves the claim. A similar argument also shows that \[\sum_{e \in p_2 - F} c_e + \sum_{S \subseteq S_2 : v \in S} y^j_S \leq \mathrm{Growth}(X_2,j).\] This completes the proof of the lemma.
\end{proof}

At any point in the algorithm, for each connected component $X$ of $F$, define the {\bf class} of $X$ to be the highest level $j$ such that it contains a terminal currently or previously active at level $j$; that is, the largest $j$ such that $X \cap (P_j \cup A_j) \neq \emptyset$.  We denote the class of $X$ as $\mathrm{Class}(X)$ and sometimes refer to it as the {\bf top level} of $X$.  Define $\mathrm{TopGrowth}(X)$ to be the maximum total dual growth of a terminal in $X$ in level $\mathrm{Class}(X)$, i.e.
\begin{align*}
\mathrm{TopGrowth}(X) & = \mathrm{Growth}(X,\mathrm{Class}(X))\\
&  = \max_{s\in X} \{\sum_{S \subseteq V: s \in S} y_S^{\mathrm{Class}(X)}\mbox{ and } s \mbox{ is a terminal} \}.
\end{align*}
For example, consider again the instance in Figure \ref{fig:alg-ex}. At the end of level $0$, we have $\mathrm{Class}(\{s_2,s_3\}) = -1$ and  and $\mathrm{TopGrowth}(\{s_2,s_3\}) = 2^{-1}$; we also have $\mathrm{Class}(\{s_1\}) = \mathrm{Class}(\{s_4\}) = 0$ and $\mathrm{TopGrowth}(\{s_1\}) = \mathrm{TopGrowth}(\{s_4\}) = 2^{0}$.

We know that $\mathrm{TopGrowth}(X) \leq 2^{\mathrm{Class}(X)}$ by the dual limit on level $\mathrm{Class}(X)$.  We now show the following, which is the technical heart of our result.

\begin{lemma} \label{cost} At any time in the execution of the algorithm, the following two invariants hold:
\begin{enumerate}
\item Every connected component $X$ of $F$ has
$$\mathrm{Account}(X) \geq 2^{\mathrm{Class}(X)} + \mathrm{TopGrowth}(X);$$
\item $\sum_{e\in F}c_e  +  \sum_{X \in F} \mathrm{Account}(X) = \sum_j\sum_S y_S^{j}$.
\end{enumerate}
\end{lemma}

Invariant 1 ensures that for a component $X$, $\mathrm{Account}(X)$ stores at least $2^j$ total potential for each level $j< \mathrm{Class}(X)$ plus the maximum total dual growth of a terminal in $X$ at the top level, which gives total potential at least $2^{\mathrm{Class}(X)-1}+2^{\mathrm{Class}(X)-2}+... = 2^{\mathrm{Class}(X)}$ plus $\mathrm{TopGrowth}(X)$.

\begin{proof} Since accounts get credited for dual growth and are debited exactly the cost of edges in $F$, invariant 2 holds at any point in the execution of the algorithm.

We now prove the first invariant by induction on the algorithm.  It is easy to see that this invariant holds when no edges have been added to $F$ since the algorithm grows dual variables in level $j$ until some active dual variable reaches its limit $2^j$; it then grows duals in next higher level.
Thus, $\mathrm{Account}(X)$ is credited $2^j$ for each level $j$ below the top level $\mathrm{Class}(X)$ while getting $\mathrm{TopGrowth}(X)$ for the top level. 

We now turn to the inductive proof of invariant 1. Suppose the invariant holds just before we add some path $p$ to $F$ at level $j$ that minimizes $\sum_{e \in p-F}c_e$. Suppose the path connects terminals $s_1, s_2 \in A_j \cup P_j$, in components $X_1$ and $X_2$, respectively, of $F$. Let $X_3$ be the component that results from adding path $p$ to $F$. Define $j_1 = \mathrm{Class}(X_1)$ and $j_2 = \mathrm{Class}(X_2)$. Our shadow algorithm merges the unused potential remaining in $\mathrm{Account}(X_1)$  and $\mathrm{Account}(X_2)$ into $\mathrm{Account}(X_3)$, and removes potential from  $\mathrm{Account}(X_3)$ to pay for the cost of building edges $p-F$. Thus, we have
\begin{align*}
  \mathrm{Account}(X_3) 
  &= \mathrm{Account}(X_1) + \mathrm{Account}(X_2) - \sum_{e \in p - F}c_e \\
  &\geq 2^{j_1} + \mathrm{TopGrowth}(X_1) + 2^{j_2} + \mathrm{TopGrowth}(X_2) - \sum_{e \in p - F}c_e,
\end{align*}
where the inequality follows from applying the inductive hypothesis to $X_1$ and $X_2$.

We suppose without loss of generality that $j_2 \geq j_1$. It is easy to see that $\mathrm{Class}(X_3) = j_2$ and $\mathrm{TopGrowth}(X_3) = \mathrm{TopGrowth}(X_2)$. Thus, to prove invariant 1, it suffices to show that 
\[2^{j_1} + \mathrm{TopGrowth}(X_1) \geq \sum_{e \in p - F}c_e.\]
Note that this inequality is a formalization of our intuition that the potential associated with the component with the smaller account is sufficient to pay for adding the path. By Lemma \ref{path-cost}, we have \[\mathrm{Growth}(X_1,j) + \mathrm{Growth}(X_2,j) \geq \sum_{e \in p-F}c_e.\]
Moreover, the limit on the dual growth at level $j$ implies that $\mathrm{Growth}(X_1,j) \leq 2^j$ and $\mathrm{Growth}(X_2,j) \leq 2^j$. There are two cases to consider: either $j_1 > j$ or $j_1 = j$. In the first case, we get $2^{j_1} \geq 2^{j+1} \geq \mathrm{Growth}(X_1,j) + \mathrm{Growth}(X_2,j)$. On the other hand, if $j_1 = j$, then $\mathrm{TopGrowth}(X_1) = \mathrm{Growth}(X_1, j)$ and $2^{j_1} \geq \mathrm{Growth}(X_2,j)$. In both cases, we have
\[2^{j_1} + \mathrm{TopGrowth}(X_1) \geq \mathrm{Growth}(X_1,j) + \mathrm{Growth}(X_2,j) \geq \sum_{e \in p - F}c_e,\]
as desired.

Therefore,  invariant 1 holds at any time during the execution of the algorithm. 
 \end{proof}

To finish the proof, we need a statement about the total value of the dual solution over all dual variables.  This proof is similar to one in Berman and Coulston \cite{BermanC97} (page 347) about collections of balls.

\begin{lemma} \label{dualvars}
Let the dual vector $y^j$ with the maximum total dual $\sum_S y_S^{j}$ be $y^{\max}$. At the end of time step $i$, we have $\sum_j \sum_S y_S^j   \leq 2(\log |R_i| + 3) \sum_S y_S^{\max}$.
\end{lemma}

\begin{proof}
Let $X^*$ be a component in $F$ of highest class and let $c = \mathrm{Class}(X^*)$. Since $X^*$ is at level $c$, there must have been a terminal $s^* \in X^*$ that reached its limit in level $c-1$, so that  $2^{c-1} = \sum_{S: s^* \in S}y^{c-1}_S \leq \sum_S y_S^{c-1}$.  Similarly, we know that each terminal $s \in R_i$ has total dual in level $j$ of $\sum_{S \subseteq V: s \in S} y_S^j \leq 2^j$, so that the total value of the dual solution $y^j$ is at most $\sum_S y^j_S \leq |R_i| \cdot 2^j$.  Let $\ell = c - 1 - \lceil \log_2 |R_i|\rceil$; that is, $\ell$ is the level $\lceil \log_2 |R_i|\rceil$ levels below $c-1$.  We claim that we can neglect the dual value coming from levels below $\ell$ because it is not more than the value of level $c-1$.  In particular,
$$\sum_{j=-\infty}^{\ell-1} \sum_S y^j_S \leq |R_i| \sum_{j=-\infty}^{\ell-1} 2^j \leq |R_i| \cdot 2^{\ell}  \leq 2^{c-1}.$$   Then 
\begin{align*}
\sum_j \sum_S y^j_S & =  \sum_{j < \ell} \sum_S y^j_S + \sum_{j = \ell}^{c} \sum_S y^j_S \\
& \leq 2 \sum_{j =\ell}^c \sum_S y^j_S \\
& \leq 2 (c-\ell+1) \sum_S y^{\max}_S \\
& \leq 2 (\log |R_i| + 3) \sum_S y^{\max}_S.
\end{align*}
\end{proof}

Now, we are ready to prove Theorem \ref{thm}.

\begin{proofof}{Theorem \ref{thm}}  By Lemma \ref{feas}, at the end of time step $i$ of the algorithm, $F$ is a feasible solution to $(IP_i)$. We have
\begin{align*}
\sum_{e\in F}c_e &\leq  \sum_j\sum_S y_S^{j}  & \mbox{by Lemma \ref{cost}}\\
&\leq 2(\log |R_i| + 3)\sum_S y_S^{\max} & \mbox{by Lemma \ref{dualvars}}\\
& \leq 2(\log |R_i| + 3)OPT_i & \mbox{by Lemma \ref{feas}}
\end{align*}
where $OPT_i$ is the optimal value of $(IP_i)$ and the last inequality follows since the value of the feasible dual solution to $(D_i)$ $y^{\max}$ is a lower bound on $OPT_i$. Therefore, our algorithm is an $O(\log |R_i|)$-competitive algorithm for the online proper constrained forest problem. Note that we have $ |R_i| \leq n$, where $n$ is the number of nodes in $G$.  The constants can be made somewhat tighter, but we omit these details for the sake of clarity.
\end{proofof}

\section{Online Network Design with Penalties}
\label{sec:pccf}

In this section, we extend the algorithm of the previous section to one in which we are allowed to violate connectivity constraints by paying a penalty.  To do this, we will use a very general form of the problem introduced by Sharma, Swamy, and Williamson \cite{SSW07}.  In the offline version of their problem, they give a
arbitrary $0$-$1$ connectivity requirement function $f:2^V \rightarrow \{0,1\} $, and a submodular and monotone penalty function $\pi: 2^{2^V} \rightarrow \mathbb{Z}_{\geq0}$.  Note that the penalty function is on {\em collections} or {\em families} of sets, which we will denote by ${\cal S}$.  In the offline problem, we must find a set of edges $F$ and a family of sets ${\cal S}$ such that for any subset $S$ of vertices, either $|F \cap \delta(S)| \geq f(S)$ or $S \in {\cal S}$.  The goal is to minimize the cost of the edges in $F$ plus the penalty $\pi({\cal S})$.  Sharma et al.\ \cite{SSW07} restrict the penalty function to have the following properties:
\begin{itemize}
	\item (Emptyset property) $\pi(\emptyset) = 0$;
	\item (Monotonicity) If ${\cal S} \subseteq {\cal T}$, then $\pi({\cal S}) \leq \pi({\cal T})$.
	\item (Submodularity) For any collections ${\cal S}$ and ${\cal T}$, $\pi({\cal S}) +
	\pi({\cal T}) \geq \pi({\cal S} \cup {\cal T}) + \pi({\cal S} \cap {\cal T})$.
	\item \label{item:f-regular-union}(Union property) For any two subsets
	$S_1$ and $S_2$, $\pi(\{S_1,S_2,S_1\cup S_2\})=\pi(\{S_1,S_2\})$.
	\item \label{item:f-regular-complement}(Complement property)
	For any subset $S\subseteq V$, $\pi(\{S,V-S\})=\pi(\{S\})$.
	\item \label{item:f-regular-inactive}(Inactivity property) For
	any subset $S\subseteq V$ with $f(S)=0$, $\pi(\{S\})=0$.
\end{itemize}
Note that the last property implies that if a set has an associated penalty, then it must require some type of connectivity.   

\newcommand{\closure}{\mathsf{closure}}

To understand what penalty arises from a given solution $F$, let ${\cal C}$ be the connected components of $(V,F)$.  We call ${\cal T}$ the {\bf closure} of a collection of sets ${\cal S}$  if ${\cal T} \supseteq {\cal S}$ and ${\cal T}$ is closed under taking unions and complements (and thus intersections and set differences as well); we denote ${\cal T}$ by $\closure({\cal S})$.  Then given a solution $F$ and its connected components ${\cal C}$, the family of sets on which we must pay a penalty is $\closure({\cal C})$.

We extend the algorithm of the previous section to an online version of the problem, which we now define. We start with a connectivity requirement function $g_0$ and penalty function $\pi_0$, where $f_0(S) = 0$ for all $S\subseteq V$ and $\pi_0(\mathcal{S}) = 0$ for all $\mathcal{S} \subseteq 2^V$. In each time step $i$, a connectivity requirement function $f_i$ and a penalty function $\pi_i$ arrive with the following properties:
\begin{enumerate}
\item $f_i:2^V \rightarrow \{0,1\} $ is a proper function,
\item $g_i(S)= \max(f_1(S),\ldots,f_i(S))$ for all $S \subseteq V$,
\item $\pi_i: 2^{2^V} \rightarrow Z_{\geq 0}$ satisfies all other properties described above with respect to $f_i$.
\end{enumerate} 
Notice that unlike \cite{SSW07}, we require that the functions $f_i$ be proper functions, so that $g_i$ is a proper function.  We also observe that because $\pi_i$ obeys all the properties described above, then so does $\sum_{k=1}^i \pi_k$ with respect to $g_i$; in particular for the inactivity property $g_i(S) = 0$ implies that $\sum_{k=1}^i \pi_i(\{S\}) = 0.$

We will call this general online problem the {\em online prize-collecting constrained forest problem}.  In our variant of the problem, we assume that any decision made to add an edge to $F$ cannot be undone in future time steps, and any decision to pay a penalty also cannot be undone, even if we end up later fulfilling the associated connectivity constraint, in a sense that we now describe. If $i$ is the current time step, and we decide to pay the penalty for a collection ${\cal S}_i$, then we pay $\pi_i({\cal S}_i)$ in this time step and all future time steps. Thus if ${\cal S}_k$ is the collection on which we decided to pay the penalty in time step $k$, the total penalty we pay in time step $i$ is $$\sum_{k=1}^i \pi_k({\cal S}_k).$$

As usual, we compare the cost of the online algorithm in each time step $i$ to the cost of an optimal solution to the offline problem at time step $i$, and the algorithm is $\alpha$-competitive if the cost of the algorithm's solution is always within a factor of $\alpha$ of the cost of the optimal offline problem.  To be specific, the cost of the algorithm's solution is the cost of the edges plus the sum of the penalties across all time steps, while if the optimal set of edges for the offline problem is $F^*$, and the corresponding set of connected components is ${\cal C}^*$, then the cost of the optimal offline solution is $\sum_{e \in F^*} c_e + \sum_{k=1}^i\pi_k(\closure({\cal C}^*)).$

One problem captured by this framework is the online version of the prize-collecting generalized Steiner tree problem given by Hajiaghayi and Jain  \cite{HJ06}.  The online prize-collecting generalized Steiner tree problem is as follows: initially we are given an undirected graph $G$, and a penalty of zero for each pair of nodes. In each time step $i$, a terminal pair $(k, l)$ arrives with a new penalty $\pi_{kl} >0$. We have a choice to either connect $k$ to $l$ or pay a penalty $\pi_{kl}$ for not connecting them.  Our goal is to find a set of edges $F$ that minimizes the sum of edge costs in $F$ plus the sum of penalties for terminal pairs that are not connected.  The function $g_i$ is the same as for the online generalized Steiner tree problem; the penalty  function $\pi_i$ in time step $i$ for a family of sets ${\mathcal S}$ in this case is the sum
of penalties of pairs which are separated by some set in ${\mathcal S}$;
that is $\pi_i({\mathcal S}) = \pi_{kl}$ if there is an $S \in {\mathcal S}$ such that $|S \cap \{k,l\}|
= 1$, and $\pi_i({\mathcal S}) = 0$ otherwise.  Sharma et al.\ \cite{SSW07} show that $\pi_i$ obeys the required properties for the offline prize-collecting constrained forest problem.  Thus by our notion of penalties above, if we decide to pay the penalty $\pi_{kl}$ at the current time step, we continue to pay $\pi_{kl}$ in all future time steps, even if $k$ and $l$ are later connected.

Another interesting special case is the (offline) prize-collecting Steiner tree problem, first defined as Bienstock, Goemans, Simchi-Levi, and Williamson \cite{BienstockGSW93}.  In the offline version of the prize-collecting Steiner tree problem, we are given an undirected graph $G=(V,E)$, edge costs $c_e \geq 0$ for all $e \in E$, a root vertex $r \in V$, and penalties $\pi_v \geq 0$ for all $v \in V$.  The goal is to find a tree $T$ spanning the root vertex that minimizes the cost of the edges in the tree plus the penalties of the vertices not spanned by the tree; that is, we want to minimize $\sum_{e \in T} c_e + \sum_{v \in V-V(T)} \pi_v$, where $V(T)$ is the set of vertices spanned by $T$.  This is equivalent to the prize-collecting generalized Steiner tree problem in which one vertex in each terminal pair $(k,l)$ is the root $r$.  We define the online prize-collecting Steiner tree problem as follows: we are given a root node $r$ in $G$, and a penalty of zero for each non-root node. In each time step $i$, a terminal $s_i \neq r$ arrives with a new penalty $\pi_i >0$. We have a choice to either connect $s_i$ to root $r$ or pay a penalty $\pi_i$ for not connecting it. Let $R_i$ be the set of terminals that have arrived by time step $i$; that is, $R_i = \{ l: \pi_l > 0\}$.  Our goal is to find a set of edges $F$ that minimizes the sum of edge costs in $F$ plus the sum of penalties.  Since the problem is a special case of the online prize-collecting generalized Steiner tree problem, it is also a special case of the online prize-collecting constrained forest problem.  It follows that if we pay the penalty $\pi_i$ for not connecting terminal $s_i$ in time step $i$, we continue to pay the penalty in later iterations even if we later connect $s_i$ to the root.

A final special case of this problem is an online version of a problem introduced by Hayrapetyan, Swamy, and Tardos \cite{HST05}.  In the offline version of this problem, we are given an undirected graph $G=(V,E)$, edge costs $c_e \geq 0$ for all $e \in E$, a root vertex $r \in V$, and a monotone submodular penalty function $h$.  The goal is to find a tree $T$ spanning the root $r$ to minimize the cost of the edges in $T$ plus $h(S)$, where $S$ is the set of vertices not spanned by $T$.  We can give an online version of the problem by assuming that at each time step $i$ we receive a monotone submodular function $h_i$; if $S_i$ is the set of vertices not spanned at the end of time step $i$, then we pay $h_i(S_i)$ in penalty for that time step, and $\sum_{k=1}^i h_k(S_k)$ overall. To fit in our framework, we define $$p({\cal S}) = \bigcup_{S \in {\cal S}: r \notin S} S \cup \bigcup_{S \in {\cal S}: r \in S} (V-S),$$ and $\pi_i({\cal S}) = h_i(p({\cal S})).$  Sharma et al.\ show that $\pi_i$ satisfies all the properties needed by the offline prize-collecting forest problem if we assume that $f_i(S) = 1$ for all $S$ and all $i$.


The integer programming formulation of the problem in the $i$th time step is
\lps & &
& \mbox{Min} & \sum_{e \in E} c_{e} x_{e} + \sum_{\mathcal{S}} \sum_{k=1}^i\pi_k(\mathcal{S})z_{\mathcal{S}} \\
(IP_i) & & & & \sum_{e \in \delta(S)}x_{e} + \sum_{S: S\in \mathcal{S}} z_{\mathcal{S}} \geq g_i(S), & \forall S\subseteq V,\\
& & & & x_e \in \set{0,1}, & \forall e \in E,\\
& & & & z_{\mathcal{S}} \in \set{0,1}, & \mathcal{S} \subseteq 2^{V}.
 \elps
The optimal solution to the integer program gives the optimal offline solution in time step $i$.
Let $(LP_i)$ denote the corresponding linear programming relaxation in which the constraints $x_e \in  \{0,1\}$ and $z_{\mathcal{S}} \in  \{0,1\}$ are replaced with $x_e \geq 0$ and $z_{\mathcal{S}} \geq 0$. The dual of this linear program, $(D_i)$, is
\lps & &
& \mbox{ Max} &  \sum_{S \subseteq V} g_i(S) y_S  \\
(D_i) & & & & \sum_{S: e \in \delta(S)} y_{S} \leq c_e, & \forall e \in E,\\
& & & & \sum_{S: S \in \mathcal{S}} y_S \leq \sum_{k=1}^i\pi_k(\mathcal{S}), & \forall \mathcal{S} \subseteq 2^{V}\\
& & & & y_S \geq 0, & \forall S \subseteq V. \elps

For dual problem $(D_i)$, call the constraints  $\sum_{S:e\in\delta(S)} y_S  \leq   c_e$ the edge cost constraints and the constraints $ \sum_{S: S \in \mathcal{S}} y_S \leq \sum_{k=1}^i\pi_k(\mathcal{S})$ the penalty constraints. A penalty constraint corresponding to a family $\mathcal{S}^j$ is {\em tight} in level $j$ if the left-hand side of the inequality is equal to the right-hand side.

We extend the algorithm of Figure \ref{alg} to give an $O(\log |R_i|)$-competitive algorithm for the online prize-collecting constrained forest problem, where $R_i$ is defined as before; namely, $R_i$ is the set of all $v \in V$ for which $g_i(\{v\})=1$. We again call the vertices in $R_i$ terminals.

Our algorithm is similar  to the algorithm in Figure \ref{alg} in how it grows dual variables, with the same conditions (1)-(3) in that algorithm in the dual growth loop, but with an additional condition (4): when a penalty constraint corresponding to a family $\mathcal{S}^j$ becomes tight in level $j$, we mark all terminals $s$ with $s \in S^j \in \mathcal{S}^j$ and mark family $\mathcal{S}^j$ to pay its penalty.  Any marked terminal becomes inactive, and any marked moat $S^j$ also becomes inactive.  Additionally, when we update moats at the bottom of the dual growth loop for level $j$, if it is the case that $g_i(S^j) = 0$ for some moat $S^j$, then we make inactive all active terminals in $S^j$.

Let $Q$ be the collection of all families marked by our algorithm at a given point in the algorithm. At the beginning of time step $i$, we unmark each family $\mathcal{S}$ in $Q$ 
and unmark all terminals contained in a set $S$ in $\mathcal{S}$ (and all moats $S$) if $S$ is a violated set for function $g_i$. 
At the end of time step $i$, our algorithm outputs $F$ and the collection of marked families $Q$.

This algorithm can be implemented in polynomial time.  The only change from the algorithm in Figure \ref{alg} is that we need to be able to check condition (4); that is, we need to find the next dual penalty constraint to go tight in level $j$ efficiently.  To do this, we can apply the algorithm described in Section 5.3 of Sharma et al.\ \cite{SSW07}, which uses submodular function minimization; we observe that since each function $\pi_i$ is submodular and monotone, then so is $\sum_{k=1}^i \pi_i$.

\begin{figure}[t]
\begin{center}
{
\algorithm{Prize-Collecting Constrained Forest Algorithm}
{
$F = \emptyset$, $\bar{F}^{j} = \emptyset$ for all $j$, and $y_S^{j} = 0$ for all $j$ and $S\subseteq V$\\
For each $\{0,1\}$-proper function $f_i$ that arrives\\
\> Update active terminals $A$, and active moats $\mathcal{M}$\\
\> Set ${\cal X}^j = \emptyset$\\
\> For $j \gets -\infty$ to $\infty$\\
\> \> (Consolidate) While there is an edge $\bar e \in F \setminus \bar F^j$\\
\> \> \> Add $\bar e$ to $\bar F^j$\\
\> \> \> While there are terminals $s_1 \in A_j$, $s_2 \in P_j$ in the same moat $S^j$\\
\> \> \> that are not connected in $F$\\
\> \> \> \> Let $p \subseteq E$ be an $s_1$-$s_2$ path in $\bar F^j$ minimizing $\sum_{e \in p-F} c_e$\\
\> \> \> \> $F \gets F \cup \{p\}$, i.e. build edges $p-F$\\
\> \> \> \> Update $A$\\
\> \> Update active moats $\mathcal{M}$\\
 \> \> (Dual growth) While there are terminals active at level $j$\\
 \> \>\> Grow uniformly all active dual variables $y_S^{j}$ until\\
 \> \>\> 1) An active $y_S^j$ reaches its limit in level $j$\\
 \> \>\> 2) An edge $e\in E$ becomes tight in level $j$, then $\bar{F}^{j} = \bar{F}^{j} \cup \{e\}$ \\
\>  \>\> 3) Two terminals $s_1 \in A_j$ and $s_2 \in A_j \cup P_j$ connect in level $j$, then\\
\>  \>\>\>      Let $p \subseteq E$ be the $s_1$-$s_2$ path of edges in $\bar{F}^j$ minimizing $\sum_{e \in p-F} c_e$\\
\>  \>\>\>      $F = F \cup \{p \}$, i.e. build edges $p-F$ \\
\>  \>\>\>     Update $A$\\
\>  \>\>(4) A penalty constraint w.r.t.  family $\mathcal{S}^j$ becomes tight in level $j$\\
 \> \>\>\> Mark all terminals $s$ with $s \in S^j \in \mathcal{S}^j$; make $s$ and $S^j$ inactive\\
\>  \>\>\> Mark family $\mathcal{S}^j$ to pay its penalties\\
\> \> \>\> Add ${\cal S}^j$ to ${\cal X}^j$\\
\>  \>\>\>     Update $A$\\
\>  \>\> Update $\mathcal{M}$\\
  Let $Q$ be the families of sets marked to pay penalties\\
 Output $F$ and $Q$
}}
\end{center}
\caption{Primal-dual algorithm for the online prize-collecting constrained forest problem}
\label{pccf}
\end{figure}

We can reuse many parts of the analysis of the main algorithm.  There are two main changes to be concerned about.  The first is that unlike the previous algorithm, it is possible for a connected component $X$ to have $g_i(X) = 1$ but have no active terminal in it, in contradiction to Lemma \ref{inactive}.  This lemma was used in Lemma \ref{active} to show that if there is an active terminal in a moat, then the corresponding dual variable is active, so that we can be assured of obtaining a feasible solution for the function $g_i$.  Now our algorithm updates moats $S^j$ so that if $g_i(S^j) = 0$, then we make the terminals in the moat inactive. We can do so because in the penalty version of the problem we are allowed to have components $X$ of $F$ with $g_i(X)=1$ as long as we pay the associated penalty.

The second main change is that we have to pay the  penalty for the families of sets $Q$ returned by the algorithm, and also the penalties from prior time steps of the algorithm.  Because we only include such a family in $Q$ when the corresponding dual penalty constraint is tight, we can charge the additional penalty to incremental increases in dual variables.  Thus the total penalty over all time steps can be charged to a single copy of the dual variables, as we will show in detail below. So we charge the costs of the edges in $F$ to one copy of the dual variables, and the penalties to another copy; this increases the competitive ratio by a constant factor, but the ratio still remains $O(\log |R_i|)$.

To prove the result, we need the following lemmas.  The first two are from Sharma et al. \cite{SSW07}.  

\begin{lemma}[Lemma 4.1, Sharma et al.\ \cite{SSW07}] \label{lem:ssw1} Let ${\cal S}$ be a family of sets, and $S$ be any set such that $g(S) = 0$. Then, for any $S_1,  S_2 \in {\cal S}$, we have $\pi({\cal S}) = \pi({\cal S} \cup \{S_1 \cup S_2\}) = \pi({\cal S} \cup (V-S_1)) = \pi({\cal S} \cup \{S\})$.
\end{lemma}

The lemma follows from the union, complement, and inactivity properties of $\pi$.

\begin{corollary} \label{cor:closure}
$\pi({\cal S}) = \pi(\closure({\cal S}))$.
\end{corollary}

\begin{lemma}[Lemma 4.2, Sharma et al.\ \cite{SSW07}] 
	\label{lem:ssw}	
	If there are two families ${\cal S}^j$ and ${\cal T}^j$ that are tight in level $j$ for the associated penalty constraints, then the family ${\cal S}^j \cup {\cal T}^j$ is also tight in level $j$ for its associated penalty constraint.
\end{lemma}

For time step $i$, let ${\cal X}^j$ be the union of all marked families from level $j$, and let ${\cal I}_i$ be the collection of all sets $S$ such that $g_i(S) = 0$.  We defer the proof of the following lemma for a moment.

\begin{lemma} \label{lem:closure}
For any connected component $C$ of $(V,F)$ during time step $i$, we have $$C \in \closure\left({\cal I}_i \cup \bigcup_j {\cal X}^j\right).$$
\end{lemma}

\begin{corollary}  \label{cor:pen} If ${\cal C}$ is the set of all connected components of $(V,F)$ at the end of time step $i$, then
	$$\pi_i(\closure({\cal C})) \leq \sum_j \pi_i({\cal X}^j).$$
\end{corollary}

\begin{proof}
By Lemma \ref{lem:closure}, for any $C \in {\cal C}$, $C \in \closure\left({\cal I}_i \cup \bigcup_j {\cal X}^j\right)$, so that $$\closure({\cal C}) \subseteq \closure\left({\cal I}_i \cup \bigcup_j {\cal X}^j\right).$$  By the monotonicity of  the penalty function $\pi_i$, 
\begin{align*}
\pi_i(\closure({\cal C}))  
 & \leq \pi_i\left(\closure\left({\cal I}_i \cup \bigcup_j {\cal X}^j\right)\right)\\
& = \pi_i\left({\cal I}_i \cup \left(\bigcup_j {\cal X}^j\right)\right) \\
& = \pi_i\left(\bigcup_j {\cal X}^j\right) \\
& \leq \sum_j \pi_i({\cal X}^j).
\end{align*}
where the first equality follows from Corollary \ref{cor:closure}, the second equality follows by Lemma \ref{lem:ssw1}, and the final inequality by the submodularity of $\pi_i$.	
\end{proof}

We can now prove the main theorem.

\begin{theorem} The algorithm in Figure \ref{pccf} gives an $O(\log |R_i|)$-competitive algorithm for the online prize-collecting constrained forest problem $(IP_i)$.
\label{thm4}
\end{theorem}

\begin{proof}
Since $\pi_i(\mathcal{S}) \geq  0$ for all $\mathcal{S} \subseteq 2^{V}$ and $i \geq 1$, each dual solution $y^j$ that is feasible at the end of time step $i$ will remain feasible at the beginning of time step $i+1$. By construction each dual solution $y^j$ is feasible for $(D_i)$, and any set of edges $F$ is feasible for $(IP_i)$ as long as we pay the associated penalty.

To bound total edge costs and penalties, we need to bound the cost of edges built by conditions (3) and incremental penalties paid by condition (4). By Lemma \ref{dualvars}, we have $\sum_{e\in F}c_e   \leq 2(\log |R_i| + 3) \sum_S y_S^{\max}$. 

We need to use another copy of the dual variables to bound the penalties.  Denote by ${\cal X}^{j,k}$ the union of all families that went tight at level $j$ in time step $k$, and by $y^{j,k}_S$ the value of the dual variable $y^j_S$ at the end of time step $k$. Let ${\cal X}^j = {\cal X}^{j,i}$ be the union of the families in level $j$ that correspond to a tight penalty constraint in the current time step.  If ${\cal C}$ is the set of connected components at the end of the time step $i$, then  penalty added is $\pi_i(\closure({\cal C}))$, which is at most
$$\sum_j \pi_i({\cal X}^j)$$
by Corollary \ref{cor:pen}.  Thus the total penalty to be paid in this time step is at most
$$\sum_{k=1}^i \sum_j \pi_k({\cal X}^{j,k}).$$
We now show by induction that this total penalty is bounded above by the sum of the dual variables; in particular, we prove that
\begin{equation} \label{eq:pen} 
\sum_{k=1}^i \sum_j \pi_k({\cal X}^{j,k})  \leq  \sum_j \sum_S y^{j}_S. 
\end{equation}
By Lemma \ref{lem:ssw}, it must be the case that $\sum_{k=1}^i \pi_k({\cal X}^j) = \sum_{S \in \mathcal{X}^j} y_S^j$, and by the feasibility of the dual solution in time step $i-1$, it is the case that $\sum_{S \in {\cal X}^j} y_S^{j,i-1} \leq  \sum_{k=1}^{i-1} \pi_k ({\cal X}^{j}).$  Thus we have that 
\begin{align*}
\sum_j \pi_i({\cal X}^j) & = \sum_{k=1}^i \sum_j \pi_k({\cal X}^j) - \sum_{k=1}^{i-1} \sum_j \pi_k({\cal X}^j)\\
& \leq \sum_j \sum_{S \in {\cal X}^j} (y^j_S - y^{j,i-1}_S)
\end{align*}  By induction 
$$\sum_{k=1}^{i-1} \sum_j \pi_k({\cal X}^{j,k})  \leq  \sum_j \sum_S y^{j,i-1}_S.$$
Thus
\begin{align*}
\sum_{k=1}^i \sum_j \pi_k({\cal X}^{j,k}) & = \sum_{k=1}^{i-1} \sum_j \pi_k({\cal X}^{j,k}) + \sum_j \pi_i({\cal X}^j) \\
& \leq \sum_j \sum_S y_S^{j,i-1} + \sum_j \sum_S (y^j_S - y^{j,i-1}_S) \\
& \leq \sum_j \sum_S y^j_S,
\end{align*}
and Inequality (\ref{eq:pen}) is shown.

Therefore, using Lemma \ref{dualvars}, we have that 
\begin{align*}
\sum_{e \in F} c_{e} + \sum_{k=1}^i \pi_k({\cal X}^{j,k})
 & \leq 2 \sum_j \sum_S y^j_S \\
&\leq  4(\log |R_i| + 3)  \sum_S y_S^{\max} \\
&\leq O(\log |R_i|) OPT_i.
\end{align*}
 \end{proof}

We now turn to the proof of Lemma \ref{lem:closure}.

\begin{proofof}{Lemma \ref{lem:closure}}
We give a proof by contradiction.  For a given time step $i$, pick the earliest point in the algorithm at which there is a component $C$ of $(V,F)$ such that $g_i(C)=1$ for the current time step $i$, there is no active vertex in $C$, and yet $C \notin \closure\left({\cal I}_i \cup \bigcup_j {\cal X}^j\right).$ Suppose the algorithm is currently in level $j$. Just prior to this point in time, $C$ must have contained an active vertex, since at the beginning of the time step $i$, any component $C$ of $F$ with $g_i(C)=1$ must contain an active vertex. Let $S^j \in {\cal S}^j$ be the level-$j$ moat containing $C$. Recall that $S^j$ is partitioned into $C$ and other components $C'$. 

First, we show that $S^j \in {\cal I}_i \cup {\cal X}^j$. There are only two possible steps in the algorithm that could cause a terminal in $C$ to become inactive when $g_i(C)=1$. The first possibility is step (4) of the dual growth phase; in this case, a penalty constraint must have gone tight for some family ${\cal S}^j$ with $S^j \in {\cal S}^j$. The second possibility is that $g_i(S^j)=0$, so when the algorithm updated moats at the bottom of the dual growth loop, it made all active terminals in the moat inactive. In both cases, we have $S^j \in {\cal I}_i \cup {\cal X}^j$.

Second, we show that every other component $C' \neq C$ in $S^j$ belongs to $\closure\left({\cal I}_i \cup  {\cal X}^j\right)$. By Lemma \ref{unique}, prior to this point in time, all the active vertices are contained in $C$. Thus, either $g_i(C')=0$ and $C' \in {\cal I}_i$, or $g_i(C')=1$. In the latter case, since $C'$ did not contain an active vertex, in order not to contradict our choice of $C$, it must be that $C' \in \closure\left({\cal I}_i \cup {\cal X}^j\right).$ 

Since $S^j$ and every other component $C' \neq C$ in $S^j$ belongs to $\closure\left({\cal I}_i \cup {\cal X}^j\right)$, it must be that $C \in \closure\left({\cal I}_i \cup {\cal X}^j\right)$. This gives the desired contradiction and concludes the proof of the lemma.
\end{proofof}

\section{Conclusion}
\label{sec:conc}

In the online generalized Steiner network problem, we are given as input an undirected graph and nonnegative edge costs, and in the $i$th time step, a pair of terminals $(s_i$,$t_i)$ arrives with a connectivity requirement $r_i$.  One must then augment the current solution so that there are at least $r_i$ edge-disjoint paths between $s_i$ and $t_i$.  It is an interesting open question whether primal-dual algorithms for the offline generalized Steiner network  design problem (such as those in \cite{WilliamsonGMV95,GoemansGPSTW94}) can be adapted to the online case as we did here for the online constrained forest problem. Gupta, Krishnaswamy, and Ravi \cite{GuptaKR09} have shown that if $R_i$ is the set of terminals that have arrived by the $i$th time step, then there is a lower bound of $\Omega(|R_i|)$ on the competitive ratio. If $r_{\max} = \max_i r_i$, Gupta et al.\ \cite{GuptaKR09} have given an $O(r_{\max} \log^3 n)$-competitive algorithm for this problem, so such an adaptation might be possible.

Another interesting question is what happens if the algorithm is allowed to remove some small number of edges from the solution as time progresses.  In particular, Gu, Gupta, and Kumar \cite{GuGK13} have shown that it is possible to have a constant competitive ratio algorithm for the online Steiner tree problem if, in addition to adding edges at each time step, it is allowed to remove a single edge per time step; this work builds on a previous algorithm of Megow, Skutella, Verschae, and Wiese \cite{MegowSVW12}.  It would be interesting to extend their algorithm to the more general types of network design considered in this paper.

For approximation algorithms and online algorithms, it is often the case that their performance is better than the theoretical worst-case analysis. Cheung \cite{Cheung15} has performed a computational study of various online algorithms for the online prize-collecting Steiner tree problem, including the algorithm presented here and the algorithm of Umboh \cite{Umboh15}. She finds that for our algorithm that the average competitive ratio is 1.848 among 40 instances with up to 400 nodes.  Umboh's algorithm has better performance, with an average competitive ratio of 1.341.

\bibliographystyle{spmpsci}
\bibliography{onlinenetdesign}

\begin{thebibliography}{10}
\providecommand{\url}[1]{{#1}}
\providecommand{\urlprefix}{URL }
\expandafter\ifx\csname urlstyle\endcsname\relax
  \providecommand{\doi}[1]{DOI~\discretionary{}{}{}#1}\else
  \providecommand{\doi}{DOI~\discretionary{}{}{}\begingroup
  \urlstyle{rm}\Url}\fi

\bibitem{AgrawalKR95}
Agrawal, A., Klein, P., Ravi, R.: When trees collide: An approximation
  algorithm for the generalized {S}teiner problem on networks.
\newblock {SIAM} Journal on Computing \textbf{24}, 440--456 (1995)

\bibitem{AwerbuchAB04}
Awerbuch, B., Azar, Y., Bartal, Y.: On-line generalized {S}teiner problem.
\newblock Theoretical Computer Science \textbf{324}, 313--324 (2004)

\bibitem{BallMMN95}
Ball, M.O., Magnanti, T.L., Monma, C.L., Nemhauser, G.L. (eds.): Network
  Models, \emph{Handbooks in Operations Research and Management Science},
  vol.~7.
\newblock Elsevier, Amsterdam, The Netherlands (1995)

\bibitem{BermanC97}
Berman, P., Coulston, C.: On-line algorithms for {S}teiner tree problems.
\newblock In: Proceedings of the 29th Annual {ACM} Symposium on Theory of
  Computing, pp. 344--353 (1997)

\bibitem{BienstockGSW93}
Bienstock, D., Goemans, M.X., Simchi-Levi, D., Williamson, D.P.: A note on the
  prize collecting traveling salesman problem.
\newblock Mathematical Programming \textbf{59}, 413--420 (1993)

\bibitem{Cheung15}
Cheung, S.S.: Offline and online facility location and network design.
\newblock Ph.D. thesis, Cornell University, School of Operations Research and
  Information Engineering (2016)

\bibitem{GoemansGPSTW94}
Goemans, M., Goldberg, A., Plotkin, S., Shmoys, D., Tardos, E., Williamson, D.:
  Improved approximation algorithms for network design problems.
\newblock In: Proceedings of the 5th {ACM-SIAM} Symposium on Discrete
  Algorithms, pp. 223--232 (1994)

\bibitem{GoemansW92}
Goemans, M.X., Williamson, D.P.: A general approximation technique for
  constrained forest problems.
\newblock In: Proceedings of the 3rd {ACM-SIAM} Symposium on Discrete
  Algorithms, pp. 307--316 (1992)

\bibitem{GoemansW95}
Goemans, M.X., Williamson, D.P.: A general approximation technique for
  constrained forest problems.
\newblock {SIAM} Journal on Computing \textbf{24}, 296--317 (1995)

\bibitem{GuGK13}
Gu, A., Gupta, A., Kumar, A.: The power of deferral: Maintaining a constant
  competitive {S}teiner tree online.
\newblock In: Proceedings of the 45th Annual {ACM} Symposium on Theory of
  Computing, pp. 525--534 (2013)

\bibitem{GuptaKR09}
Gupta, A., Krishnaswamy, R., Ravi, R.: Online and stochastic survivable network
  design.
\newblock In: Proceedings of the 41st Annual {ACM} Symposium on Theory of
  Computing, pp. 685--694 (2009)

\bibitem{HJ06}
Hajiaghayi, M., Jain, K.: Prize-collecting generalized {S}teiner tree problem
  via a new approach of primal-dual schema.
\newblock In: Proceedings of the 17th {ACM-SIAM} Symposium on Discrete
  Algorithms, pp. 631--640 (2006)

\bibitem{HajiaghayiLP13}
Hajiaghayi, M., Liaghat, V., Panigrahi, D.: Online node-weighted {S}teiner
  forest and extensions via disk paintings.
\newblock In: Proceedings of the 54th Annual Symposium on Foundations of
  Computer Science, pp. 558--567 (2013)

\bibitem{HajiaghayiLP14}
Hajiaghayi, M., Liaghat, V., Panigrahi, D.: Near-optimal online algorithms for
  prize-collecting {S}teiner problems.
\newblock In: J.~Esparza, P.~Fraigniaud, T.~Husfeldt, E.~Koutsoupias (eds.)
  Automata, Languages, and Programming, 41st International Colloquium, ICALP
  2014, \emph{Lecture Notes in Computer Science}, vol. 8572, pp. 576--587.
  Springer (2014)

\bibitem{HST05}
Hayrapetyan, A., Swamy, C., Tardos, {\'E}.: Network design for information
  networks.
\newblock In: Proceedings of the 16th {ACM-SIAM} Symposium on Discrete
  Algorithms, pp. 933--942 (2005)

\bibitem{ImaseW91}
Imase, M., Waxman, B.M.: Dynamic {S}teiner tree problem.
\newblock {SIAM} Journal on Discrete Mathematics \textbf{4}, 369--384 (1991)

\bibitem{JohnsonMP00}
Johnson, D.S., Minkoff, M., Phillips, S.: The prize collecting {S}teiner tree
  problem: theory and practice.
\newblock In: Proceedings of the 11th {ACM-SIAM} Symposium on Discrete
  Algorithms, pp. 760--769 (2000)

\bibitem{Karp72}
Karp, R.M.: Reducibility among combinatorial problems.
\newblock In: R.~Miller, J.~Thatcher (eds.) Complexity of Computer
  Computations, pp. 85--103. Plenum Press, New York, NY (1972)

\bibitem{MegowSVW12}
Megow, N., Skutella, M., Verschae, J., Wiese, A.: The power of recourse for
  online {MST} and {TSP}.
\newblock In: A.~Czumaj, K.~Mehlhorn, A.M. Pitts, R.~Wattenhofer (eds.)
  Automata, Languages, and Programming, no. 7391 in Lecture Notes in Computer
  Science, pp. 689--700. Springer (2012)

\bibitem{QianW11}
Qian, J., Williamson, D.P.: An ${O}(\log n)$-competitive algorithm for online
  constrained forest problems.
\newblock In: L.~Aceto, M.~Henzinger, J.~Sgall (eds.) Automata, Languages, and
  Programming, no. 6755 in Lecture Notes in Computer Science, pp. 37--48.
  Springer, Berlin (2011)

\bibitem{SSW07}
Sharma, Y., Swamy, C., Williamson, D.P.: Approximation algorithms for prize
  collecting forest problems with submodular penalty functions.
\newblock In: Proceedings of the 18th {ACM-SIAM} Symposium on Discrete
  Algorithms, pp. 1275--1284 (2007)

\bibitem{Umboh15}
Umboh, S.: Online network design algorithms via hierarchical decompositions.
\newblock In: Proceedings of the 26th Annual {ACM-SIAM} Symposium on Discrete
  Algorithms, pp. 1373--1387 (2015)

\bibitem{WilliamsonGMV95}
Williamson, D.P., Goemans, M.X., Mihail, M., Vazirani, V.V.: A primal-dual
  approximation algorithm for generalized {S}teiner network problems.
\newblock Combinatorica \textbf{15}, 435--454 (1995)

\end{thebibliography}

\end{document}